\documentclass[12pt,  draftclsnofoot,onecolumn,article,romanappendices]{IEEEtran}
\usepackage{setspace}

\usepackage{cite}
\ifCLASSINFOpdf
   \usepackage[pdftex]{graphicx}
\else
  \usepackage[dvips]{graphicx}
\fi
\usepackage[cmex10]{amsmath}
\usepackage{amssymb}
\usepackage{amsfonts}
\usepackage{amsthm}
\usepackage{array}
\usepackage{dsfont}
\usepackage{amsbsy}
\usepackage{epstopdf}

\usepackage{hyperref}

\newtheorem{lemma}{Lemma}
\newtheorem{proposition}{Proposition}
\newtheorem{remark}{Remark}
\newtheorem{corollary}{Corollary}
\long\def\symbolfootnote[#1]#2{\begingroup%
\def\thefootnote{\fnsymbol{footnote}}\footnote[#1]{#2}\endgroup}
\newtheorem{theorem}{Theorem}

\begin{document}

\title{Joint Source-Channel Coding with Time-Varying Channel and Side-Information}

\author{
  \IEEEauthorblockN{I\~naki Estella Aguerri and Deniz G\"{u}nd\"{u}z\footnote{
  I\~naki Estella Aguerri is with the Mathematical and Algorithmic Sciences Lab, France Research Center, Huawei Technologies Co. Ltd., Boulogne-Billancourt, France. E-mail: {\tt inaki.estella@huawei.com \tt}.
  Deniz G\"{u}nd\"{u}z is with the Department of Electrical and Electronic Engineer at Imperial College London, London, UK. E-mail: {\tt  d.gunduz}{\tt @imperial.ac.uk}\\
    Part of the reseach was done during the Ph.d. studies of I\~naki Estella Aguerri at Imperial College London.  This paper was presented in part at the  IEEE International Conference on
Communications, Kyoto, Japan, Jun. 2011  \cite{estella2011icc}, at the IEEE International Symposium
on Information Theory, St. Petersburg, Russia, Aug. 2011 \cite{ Estella2011DistExponent}, and at the  IEEE International Symposium on
Information Theory, Istanbul, Turkey, Jul. 2013 \cite{Estella2013SystematicLossy}.}\\}
}

\maketitle
\vspace{-1cm}
\begin{abstract}Transmission of a Gaussian source over a time-varying Gaussian channel is studied in the presence of time-varying correlated side information at the receiver. A block fading model is considered for both the channel and the side information, whose states are assumed to be known only at the receiver. The optimality of separate source and channel coding in terms of  average end-to-end distortion is shown when the channel is static while the side information state follows a discrete or a continuous and quasiconcave distribution. When both the channel and side information states are time-varying, separate source and channel coding is suboptimal in general. A partially informed encoder lower bound is studied by providing the  channel state information to the encoder. Several achievable transmission schemes are proposed based on uncoded transmission, separate source and channel coding, joint decoding as well as hybrid digital-analog transmission. Uncoded transmission is shown to be optimal for a class of continuous and quasiconcave side information state distributions, while the channel gain may have an arbitrary distribution. To the best of our knowledge, this is the first example in which the uncoded transmission achieves the optimal performance thanks to the time-varying nature of the states, while it is suboptimal in the static version of the same problem. Then, the optimal \emph{distortion exponent}, that quantifies the exponential decay rate of the expected distortion in the high SNR regime, is characterized for Nakagami distributed channel and side information states, and it is shown to be achieved by hybrid digital-analog and joint decoding schemes in certain cases, illustrating the suboptimality of pure digital or analog transmission in general.
\end{abstract}

\begin{IEEEkeywords} Distortion exponent, fading channel and side information, side information diversity, uncoded transmission, hybrid digital-analog transmission, joint decoding, joint source-channel coding. \end{IEEEkeywords}

\IEEEpeerreviewmaketitle
\section{Introduction}

Many common applications, such as multimedia  transmission over cellular networks, or the accumulation of local sensor measurements at a fusion center, require the transmission of a continuous amplitude source signal over a wireless fading channel, to be reconstructed with the minimum possible average distortion at the destination. Depending on the application layer requirements, additional delay constraints might be imposed on the system. For example, in video streaming or voice transmission, the source signal has to be reconstructed within a certain deadline. Moreover, in many practical scenarios, in addition to the received signal, the destination might have access to additional side information correlated with the source signal. This correlated side information might be obtained either from other transmitters in the network, or through the own sensing devices of the destination. While current protocols do not exploit this extra information, theoretical benefits of correlated side information are well known \cite{wyner1978rate}. We model this practical communication scenario as a joint source-channel coding problem of transmitting a Gaussian source over a time-varying Gaussian channel with the minimum average end-to-end distortion in the presence of time-varying correlated side information at the receiver.  We consider a block fading model for the states of both the channel and the side information, which are assumed to be known perfectly at the receiver.

When both the channel and the side information are static, Shannon's separation theorem applies \cite{Shamai:IT:98}, and the optimal performance is achieved by separate source and channel coding; that is, the concatenation of an optimal Wyner-Ziv source code \cite{wyner1978rate} with an optimal capacity achieving channel code. However, under strict delay constraints, if the channel and the side information are time-varying, and the channel state information (CSI) is available only at the receiver, the transmitter cannot know the optimal source and channel coding rates, and the separation theorem
fails. In order to have a good performance on average, the transmitter has to adapt to the time-varying nature of both the channel and the side information without knowing their realizations.

Strategies based on separate source and channel coding suffer from the threshold effect and do not adapt well to the uncertainties of the channel \cite{mittal2002hybrid}. On the other hand, uncoded (analog) transmission is a simple joint source-channel coding scheme robust to signal-to-noise (SNR) mismatch, and does not suffer from the threshold effect. In Gaussian point-to-point channels, uncoded transmission is an alternative optimal scheme in the absence of side information \cite{Goblick:IT:65}, \cite{Gastpar:IT:03}. However, it becomes suboptimal in the presence of correlated side information. In broadcasting and relaying scenarios where multiple users with different channel and side information qualities are present, a purely digital coding scheme based on joint decoding of the channel and source codewords, is shown to exhibit improved robustness to the threshold effect, and to achieve the optimal or superior performance \cite{tuncel2006slepian, Nayak2010DigSchemes, Gunduz2013:Relay}. 
In \cite{wilson2010joint} a hybrid digital-analog coding scheme, called HDA, is proposed and shown to be robust to the threshold effect, and unlike uncoded transmission, HDA is optimal even in the presence of correlated side information at the receiver. 
Hybrid digital analog transmission is also shown to outperform separate source and channel coding and uncoded transmission in certain static setups, such as the transmission of a Gaussian source in the presence of correlated interference \cite{Huang2012CorrInt},\cite{ Varasteh2012:SP}, the transmission of a bivariate Gaussian over a multiple access channel \cite{lapidoth2010sending} or an interference channel \cite{Estella2015ITW:JSCCIC}, or to achieve the optimal distortion in the transmission of a bivariate Gaussian source over a broadcast channel \cite{Tian2011BivariateBroadcast}.

Characterization of the optimal expected distortion  in the absence of time-varying side information has received a lot of interest in recent years\cite{ng2007minimumLayered,tian2008SuccRefinement,gunduz2008joint,Caire2007hybrid,bhattad2008distortion}.
 Despite the ongoing efforts, the optimal performance remains an open  problem. The expected distortion in this model is studied using multi-layer source codes concatenated superposition coding schemes \cite{ng2007minimumLayered,tian2008SuccRefinement}. More conclusive results on this problem have been obtained by focusing on the high SNR behavior of the expected distortion. The SNR exponent of the expected distortion, called the \emph{distortion exponent}, is characterized in the multi-antenna setup in certain regimes in \cite{gunduz2008joint},\cite{Caire2007hybrid} and \cite{bhattad2008distortion}, and it is shown that multi-layer source and channel codes, or hybrid digital-analog coding schemes are needed to achieve the optimal distortion exponent.

The pure source coding version of our problem, in which the channel is considered as an error-free constant-rate link, is studied in \cite{ng2007minimum}, and it is shown that, contrary to the channel coding problem,  when the side information follows a continuous quasiconcave fading distribution, a single layer source code suffices to achieve the optimal performance. Recently, the joint source channel coding problem has also been considered in \cite{Zhao2010ImpactSideInfo} and \cite{Koken2013LossyHDA}. In \cite{Zhao2010ImpactSideInfo}, the distortion exponent for separate source and channel coding is derived when the side information sequence has two states, the side information average gain does not increase with the SNR, and the channel has Rayleigh fading. In \cite{Koken2013LossyHDA}, HDA and joint decoding schemes are considered, and their performance is studied  using the distortion loss, which quantifies the loss with respect to a fully informed encoder that perfectly knows the channels and side information states.

In this paper, we consider the joint source-channel coding problem both in the finite and high SNR regimes. We first show the optimality of separate source and channel coding when the channel is static. Leveraging on this result and by providing the encoder with the channel state information, we derive a lower bound on the expected distortion. We then study achievable schemes based on uncoded transmission, separate source and channel coding (SSCC), joint decoding (JDS), as well as hybrid digital-analog transmission (S-HDA) and compare the performance of these schemes with the lower bound. We show that uncoded transmission meets the  lower bound when the side information fading state belongs to a certain class of continuous quasiconcave distributions, while separate source and channel coding is suboptimal. This class includes monotonically decreasing functions, which occur, for example, under Rayleigh fading. To the best of our knowledge, this is the first result showing the optimality of uncoded transmission in a fading channel scenario while it would be suboptimal in the static case. Then, we show that JDS always outperforms SSCC, and we numerically show that S-HDA performs very close to the proposed lower bound, although in general no particular scheme outperforms the others at all conditions.

Next, we obtain the distortion exponent corresponding to the proposed upper and lower bounds for Nakagami distributed channel and side information states. We parameterize the uncertainty by the shape parameter, denoted by $L_c$ for the channel and $L_s$ for the side information. For $L_c\geq 1$, we characterize the optimal distortion exponent and show that it is achieved by S-HDA, in line with the numerical results. For $L_c<1$, we show that JDS achieves the optimal distortion exponent in certain regimes, while S-HDA is suboptimal. However, as $L_s$ increases, the performance of JDS saturates, and eventually becomes worse than S-HDA, whose distortion exponent converges to the upper bound.

We will use the following notation in the rest of the paper. We denote random variables with upper-case letters, e.g., $X$, their realizations with lower-case letters, e.g., $x$, and the sets with calligraphic letters $\mathcal{A}$. We denote $\mathrm{E}_{X}[\cdot]$ as the expectation
with respect to $X$, and $\mathrm{E}_{\mathcal{A}}[\cdot]$ as the expectation
over the set $\mathcal{A}$. We denote by $\mathds{R}_n^+$ the set of positive real numbers, and by $\mathds{R}_n^{++}$ the set of strictly positive real numbers in the $n$-dimensional Euclidean space $\mathds{R}^n$, respectively. We define $(x)^+=\max\{0,x\}$. Given two functions $f(x)$ and $g(x)$, we use $f(x)\doteq g(x)$ to denote the exponential equality $\lim_{x\rightarrow\infty}\frac{\log f(x)}{\log g(x)}=1$, while  $\stackrel{.}{\geq}$ and $\stackrel{.}{\leq}$ are defined similarly.

The rest of the paper is organized as follows: in Section \ref{sec:SysMod} we introduce the system model. In Section \ref{sec:PreliminaryResults} we review some of the related previous results, and characterize the optimal performance for a static channel. In Section \ref{sec:UpperAndLowerBounds} we propose upper and lower bounds on the performance. In Section \ref{sec:OptimalityUncoded} we prove the optimality of uncoded transmission under certain side information fading distributions. In Section
\ref{sec:FiniteSNRResults} we provide numerical results for the finite SNR regime, while in  Section \ref{sec:HighSnrAnalysis} we consider a high SNR analysis, and characterize the optimal distortion exponent.  Finally, in Section \ref{sec:Conclusions} we  conclude the paper.

\section{System Model}\label{sec:SysMod}
\begin{figure}
\centering
\includegraphics[width=0.7\textwidth]{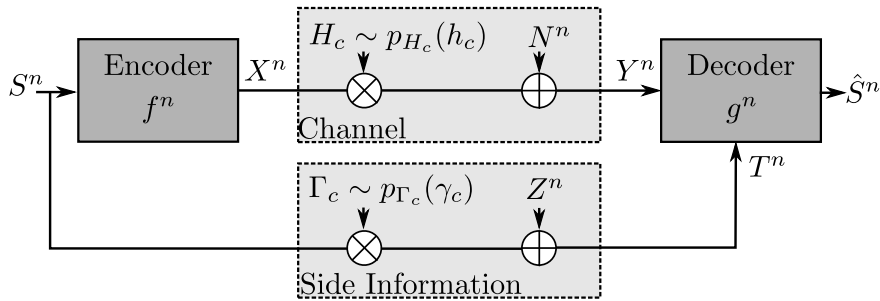}
\caption{Block diagram of the joint source-channel
coding problem with fading channel and side information.}
\label{fig:Model}
\end{figure}
We consider the transmission of a random source sequence $S^n$ of independent and
identically distributed (i.i.d.) entries form a zero mean, unit variance real
Gaussian distribution, i.e., $S_i\sim \mathcal{N}(0,1)$,
over a time-varying channel (see Fig. \ref{fig:Model}).  An encoder $f^n:\mathds{R}^n\rightarrow\mathds{R}^{n}$ maps the source sequence
$S^n$ to the channel input, $X^n\in \mathds{R}^{n}$, i.e., $x^n=f^n(s^n)$, while satisfying an average power constraint:
$\frac{1}{n}\sum_{i=1}^n\mathrm{E}[X_i^2]\leq1$. The block-fading channel is given by
 \begin{IEEEeqnarray}{rCl}
 Y^n = H_c X^n + N^n,\label{eq:ChannelModel}
 \end{IEEEeqnarray}
 where $H_c\in \mathds{R}$ is the channel fading state with probability density function (pdf) $p_{H_c}(h_c)$, and $N^n$ is the additive white Gaussian
noise $N_i\sim\mathcal{N}(0,1)$, $\forall i$.

In addition, there is an orthogonal block-fading side information channel connecting the source to the destination, which provides an uncoded noisy version of the source sequence to the destination. This second channel models the time-varying correlated side-information at the destination. Similarly to the communication channel, we model this side information channel as a memoryless block fading
 channel given by
\begin{IEEEeqnarray}{rCl}\label{eq:SImodel}
T^n=\Gamma_c S^n + Z^n,
\end{IEEEeqnarray} where  $\Gamma_c\in
\mathds{R}$ is the side information fading state with
pdf $p_{\Gamma_c}(\gamma_c)$, $X^n$
is the uncoded channel input, and $Z^n$ is the additive white
Gaussian noise, i.e., $Z_i\sim\mathcal{N}(0,1)$, $\forall i$.

We define ${H}\triangleq H_c^2\in \mathds{R}^+$ and $\Gamma\triangleq \Gamma_c^2\in \mathds{R}^+$ as the instantaneous \emph{channel gain} and the instantaneous \emph{side information gain}, with pdfs $p_{{H}}(h)$ and $p_{\Gamma}(\gamma )$, respectively.

We assume a stringent delay constraint that imposes each source block of $n$ source samples to be transmitted over one channel block, consisting of
$n$ channel uses. Both the channel and side information states, $H_c$ and $\Gamma_c$, are assumed
to be constant, with values $h_c$ and $\gamma_c$, respectively, for the
duration of one channel block, and independent among different blocks. The
channel and side information state realizations $h_c$ and $\gamma_c$ are
assumed to be known at the receiver, while the encoder is only aware
of their distributions.

The decoder reconstructs the source sequence from the
 channel output $Y^n$, the side information sequence $T^n$, and the channel and side information states $h_c$ and $\gamma_c$, using a mapping
$g^n: \!\mathds{R}^{n} \!\times \!\mathds{R}^n \!\times\!\mathds{R}\!\times\!
\mathds{R} \!\!\rightarrow \!\!\mathds{R}^n$, where
$\hat{S}^n\!=g^n(Y^n,T^n,h_c,\gamma_c)$.

For given channel and side information distributions, we are interested in
characterizing the minimum \emph{expected distortion}, $\mathrm{E}[D]$, where
the quadratic distortion between the source sequence and the
reconstruction is given by
\begin{IEEEeqnarray}{rCl}
D\triangleq\frac{1}{n}\sum^n_{i=1} (X_i-\hat{X}_i)^2.
\end{IEEEeqnarray}
The expectation is taken with respect to the source, channel and side information states,
and the noise distributions. The minimum expected distortion can be
expressed as
\begin{IEEEeqnarray}{rCl}\label{eq:EDOptimal}
ED^*\triangleq\lim_{n\rightarrow\infty}\min_{f^n,g^n}\mathrm{E}[D].
\end{IEEEeqnarray}


\section{Preliminary Results}\label{sec:PreliminaryResults}
We first review some of the existing results in the literature for the source coding version of the problem under consideration, in which the fading channel is substituted by an error-free channel of finite capacity. We then focus on the scenario in which the channel is noisy but static, i.e., the channel gain is constant and known both at the encoder and the decoder. We show that separate source and channel coding is optimal in this scenario.
\subsection{Background: Lossy Source Coding with Fading Side
Information}\label{sec:SourceCoding}

In the source-coding version of our problem the fading
channel is substituted by an error-free channel of rate
$R$ and a time-varying side information sequence $T^n$ is available at the destination \cite{ng2007minimum}. Here we briefly review the
results of \cite{ng2007minimum} which will be used later in the
paper.

Let the distribution $p_{\Gamma}(\gamma )$ be discrete with $M$
states $\gamma_{1}\leq\cdots\leq\gamma_{M} $ with probabilities
$\text{Pr}[\Gamma=\gamma_{i}]=p_i$. We define the side
information sequence available at the decoder when the realization
of the side information fading gain is $\gamma_{si}$ as $T_{i,1}^{n}\triangleq \sqrt{\gamma_{i}}
S^n+Z^n$ \footnote{To avoid confusion in the indexing, we use $T_{i,1}^{n}\triangleq[T_{i,1},...,T_{i,n}]$ to denote all the elements $T_{i,j}$, $j=1,...,n$, for the side information sequence in the $i$-th  state.}. Note that the side information has a degraded structure, characterized by the Markov chain
\begin{IEEEeqnarray}{rCl}
T_{1,j}-\cdots-T_{M-1,j}-T_{M,j}-S_j, \quad j=1,...,n.
\end{IEEEeqnarray}
 This is equivalent to the Heegard-Berger source coding problem with degraded
side information \cite{Berger1985SideInfMayBeAbsent}, in which an encoder is connected by an error-free channel of rate $R$ to $M$ receivers, and receiver $i$ has access to side information $T^n_{i,1}$.
%
It is shown in \cite{ng2007minimum} that the optimal rate allocation  can be obtained as the solution to a convex optimization problem.

%
When $p_{\Gamma}(\gamma)$ is
continuous and quasiconcave\footnote{A function $g(x)$ is
quasiconcave if its supersets $\{x|g(x)\geq \alpha\}$ are convex for
all $\alpha$.}, the optimal expected distortion is achieved by
single-layer rate allocation targeting a single side information state $\bar{\gamma} $
\cite{ng2007minimum}.
Then, the optimal expected distortion is given by
\begin{IEEEeqnarray}{rCl}\label{eq:SingleLayerCode}
ED_{Q}^*(R)\!=\!\int_{0}^{\bar{\gamma} }\!\frac{p_{\Gamma}(\gamma)}{1+\gamma}d\gamma\!+\!\int_{\bar{\gamma} }^{\infty}\!\frac{p_{\Gamma}(\gamma)}{(\bar{\gamma} +1)2^{2R}+\gamma-\bar{\gamma} }d\gamma,
\end{IEEEeqnarray}
where $\bar{\gamma} $ minimizing (\ref{eq:SingleLayerCode}) is determined as follows: Let a super-level set
be defined as
$[\gamma_l(\alpha),\gamma_r(\alpha)]\triangleq\{\gamma|
p_{\Gamma}(\gamma)\geq \alpha\}$. Then, $\bar{\gamma} $ is defined as the left endpoint of the super-level set induced by $\alpha^{*}$, i.e.,
$\bar{\gamma} =\gamma_l(\alpha^*)$, where $\alpha^*\in[0,\max
p_{\Gamma}(\gamma)]$ is found by solving the equation
\begin{IEEEeqnarray}{rCl}\label{eq:OptGamma}
\int_{\gamma_l(\alpha^*)}^{\infty}\frac{p_{\Gamma}(\gamma)-\alpha^*}{((1+\gamma_{l}(\alpha^*))2^{2R}+\gamma-\gamma_{l}(\alpha^*))^2}d\gamma=0.
\end{IEEEeqnarray}

If the side information state is Rayleigh distributed, the side information gain $\Gamma$ is exponentially distributed. Then it can be seen
that $\bar{\gamma} =0$ and the optimal expected distortion becomes
\begin{IEEEeqnarray}{rCl}
ED^*_{Ray}(R)=\frac{1}{\mathrm{E}[\Gamma]}e^{\frac{2^{2R}}{\mathrm{E}[\Gamma]}}E_{1}\left(\frac{2^{2R}}{\mathrm{E}[\Gamma]}\right),
\end{IEEEeqnarray}
where $E_1(x)\triangleq \int_{x}^{\infty}t^{-1}e^{-t}dt$ is the
exponential integral \cite{ng2007minimum}.

Results in our paper are valid for discrete, i.e., finite or countable, number of
states $\gamma_i$, as well as continuous quasiconcave side information distributions. To
unify these results, we define the function $ED^*_s(R)$ as the minimum expected
distortion in the source coding problem for these three setups.

\subsection{Static Channel and Fading Side Information}\label{sec:KnownState}

Consider a static noisy channel from $X^n$ to $Y^n$ of capacity $\mathrm{C}$. The side information is still block-fading as in (\ref{eq:SImodel}) with distribution $p_{\Gamma}(\gamma)$. Note that it is a joint source-channel coding generalization of the source coding problem in Section \ref{sec:SourceCoding}. We
denote the minimum expected distortion in the case of a static channel by $ED^*_{sta}$.
%
%
%
Optimality of separate source and channel coding in this scenario can be proven for discrete or continuous quasiconcave side information distributions. This reduces the problem to the source coding problem of Section \ref{sec:SourceCoding} with $R=\mathrm{C}$.

\begin{theorem}\label{th:separation}
 Assume that the  channel is static with capacity $\mathrm{C}$. When the side information gain $\Gamma$ has a discrete or a continuous quasiconcave pdf, $p_{\Gamma}(\gamma)$, the minimum expected
distortion, $ED^*_{sta}$, is achieved by separate source and channel
coding, and is given by
\begin{IEEEeqnarray}{rCl}
ED^*_{sta}= ED_s^*(\mathrm{C}).
\end{IEEEeqnarray}
\end{theorem}
\begin{proof}
 See Appendix \ref{app:separation} for a sketch of the proof.
\end{proof}

\section{Upper and Lower Bounds}\label{sec:UpperAndLowerBounds}

In this section we return to the problem presented in Section
\ref{sec:SysMod} in which both the channel and the side information gains are block-fading. We construct two lower bounds on $ED^*$. The first one is obtained by informing the
encoder with both the channel and side information states ${H}$ and $\Gamma$. Then, we construct a tighter lower bound by informing the encoder
only with the channel state ${H}$. Next, we obtain upper bounds on $ED^*$ based on various achievable schemes. Comparison of the proposed upper and lower bounds in different regimes of operation is relegated to later sections.

\subsection{Informed Encoder Lower Bound}\label{ssec:InformedBound}
A trivial lower bound on $ED^*$, called the \emph{informed encoder lower bound}, is obtained by providing the
encoder with the instantaneous states of the channel and the side information. At each realization, the problem reduces to
the systematic model considered in \cite{Shamai:IT:98} (see also
\cite{Steinberg2006HierarchicalJointSourceChannel}), for which  separation holds. For states $(h,\gamma)$, the
optimal distortion is given by
$D_{\text{inf}}(h,\gamma)\triangleq(1+h)^{-1}(1+\gamma)^{-1}$. Hence, the informed encoder lower bound is given by
\begin{IEEEeqnarray}{rCl}\label{eq:EDInformedBound}
ED_{\text{inf}}^*=\mathrm{E}_{{H},\Gamma}[D_{\text{inf}}({H},\Gamma)].
\end{IEEEeqnarray}
\subsection{Partially Informed Encoder  Lower Bound}\label{ssec:PartiallyInformedBound}
We can obtain a tighter lower bound called the \emph{partially informed encoder lower bound}, by providing the encoder only with the channel realization $h$. For a given
channel realization $h$, the setup reduces to the one
considered in Section \ref{sec:KnownState}, and for a discrete or continuous quasiconcave $p_{\Gamma}(\gamma)$, separation
 applies for each channel realization.
\begin{lemma}\label{lem:GenlowerBoundSISO}
If $p_{\Gamma}(\gamma)$ is discrete or continuous quasiconcave, the minimum expected distortion is lower bounded by
\begin{IEEEeqnarray}{rCl}
ED^*_{pi}\triangleq \mathrm{E}_{{H}}[ED_{s}^*(\mathcal{C}({H}))]\label{eq:GenlowerBoundSISO},
\end{IEEEeqnarray}
where $\mathcal{C}(h)\triangleq\frac{1}{2}\log(1+h)$ is the capacity of the
 channel for a given realization $h=h_c^2$.
\end{lemma}

Providing only the side information state to the encoder does not lead to a tight computable lower bound, since the optimality of separate source and channel coding does not hold in this case.
Although the partially informed encoder lower bound is tighter, we will include the informed encoder bound in our results, as it provides a benchmark for the performance when both channel and side information states are available at the transmitter, and sheds light on the value of the CSI feedback for this joint source-channel coding problem.

\subsection{Uncoded Transmission}\label{sec:UncTr}

Uncoded transmission is a memoryless zero-delay transmission
scheme in which the channel input $X_i$ is generated by scaling the
source signal $S_i$ while satisfying the power constraint. In our model
both the source variance and power constraint of the encoder are
$1$, and hence, no scaling is needed, i.e., $X_i=S_i$. The received signal from the channel is then given by
$Y_i= h_c S_i+Z_i,$ $i=1,...,n.$
The receiver reconstructs each component with a minimum mean-squared error (MMSE)\footnote{
For available data vector $\mathbf{A}\sim \mathcal{N}(0,\mathbf{C}_a)$, the MMSE in estimating  the Gaussian vector $\mathbf{X}\sim \mathcal{N}(0, \mathbf{C}_x)$ is achieved with the estimator $\hat{\mathbf{X}}=\mathrm{E}[\mathbf{X}|\mathbf{A}]$ and is given by $D_{\text{MMSE}}\triangleq(\mathbf{C}_x+\mathbf{C}_{xa}^H\mathbf{C}_a^{-1}\mathbf{C}_{xa})^{-1}$, where $\mathbf{C}_{xa}\triangleq \mathrm{E}[\mathbf{AX}^H]$ \cite{elGamal:book}.}
 estimator using both the channel output and the side information sequence, i.e., $\hat{S}_i=\mathrm{E}[S_i|Y_i,T_i]$, $i=1,...,n$. The
distortion of source component $S_i$ for channel and side information gains $h$ and $\gamma$ is given by $D_{u}(h,\gamma)\triangleq(1+h+\gamma)^{-1}$.
The achievable average distortion with uncoded transmission is given by
\begin{IEEEeqnarray}{rCl}
ED_u&=&\mathrm{E}_{{H},\Gamma}[D_{u}({H},\Gamma)].
\end{IEEEeqnarray}


\subsection{Separate Source and Channel Coding (SSCC)}\label{ssec:SeparateBinning}
In SSCC a single
layer Wyner-Ziv source code is followed by a channel code. Note that due to the lack of state information at the transmitter the rates of the source and channel codes
are independent of the channel and side information states.

The quantization codebook consists of $2^{n(R_c+R_s)}$ length-$n$
codewords, $W^n(i)$, $i=1,...,2^{n(R_c+R_s)}$, generated through a
`test channel' $W=S+Q$, where
$Q\sim\mathcal{N}(0,\sigma_Q^2)$ is independent of $S$. The
quantization noise variance is chosen such that $R_s+R_c=I(S;W)+\epsilon$, for an arbitrarily small $\epsilon>0$,
i.e., $\sigma^2_Q=(2^{2(R_s+R_c-\epsilon)}-1)^{-1}$. The generated quantization
codewords are then uniformly distributed into $2^{nR_c}$ bins. Each bin index $s$ is assigned to an independent Gaussian channel codeword $X^n(s)$, $s\in[1,..., 2^{nR_c}]$, generated i.i.d.  with $X\sim\mathcal{N}(0,1)$.
%
Given source realization $S^n$, the encoder searches for a
codeword $W^n(i)$  jointly
typical\footnote{For definition and properties of typicality and jointly typical random variables we refer the reader to \cite{Cover:book}.} with $s^n$, and transmits the corresponding channel codeword $X^n(s)$, where $s$ is the bin index of $W^n(i)$.

At reception, the decoder tries to recover the bin index $s$ using the  channel output $Y^n$, and then looks for a quantization codeword within the estimated bin, that is jointly typical with the side information
sequence $T^n$. If the quantization codeword $W^n$ is successfully decoded, then
$\hat{S}^n$ is reconstructed with an optimal MMSE estimator as $\hat{S}_i=\mathrm{E}[S_i|T_i,W_i]$ for
$i=1,...,n$. 
%
An outage is declared whenever, due to the randomness of the
channel or the side information, the quantization codebook
cannot be correctly decoded. The outage event is given by
\begin{IEEEeqnarray}{rCl}\label{eq:OutageSetsBinning}
\mathcal{O}_{s}&\triangleq&\{(h,\gamma):R_c\geq I(X;Y)
\text{ or }R_c\leq I(S;W|T)\},\nonumber
\end{IEEEeqnarray}
where $I(S;W|T)=\frac{1}{2}\log\left(1+(2^{2(R_s+R_c+\epsilon)}-1)/(\gamma+1)\right)$ and  $I(X;Y)=\frac{1}{2}\log(1+h)$.

For a quantization rate is $R$, if the quantization codeword is decoded correctly, and the side information state is $\gamma$, the achieved distortion is
\begin{IEEEeqnarray}{rCl}\label{eq:DigiDist}
 D_d(R,\gamma)\triangleq(\gamma+2^{2R})^{-1},
\end{IEEEeqnarray}

If an outage occurs, only the
side information sequence is used to estimate the source, and we have $\hat{S}_i=\mathrm{E}[S_i|T_i]$, and the achievable distortion is given by
$D_d(0,\gamma)$. Then, the expected distortion of SSCC is given by
\begin{IEEEeqnarray}{rCl}\label{eq:SepED}
ED_{s}(R_s,R_c)&=&\mathrm{E}_{\mathcal{O}^c_{s}}[D_d(R_s+R_c,\Gamma)]+\mathrm{E}_{\mathcal{O}_{s}}[D_d(0,\Gamma)]\nonumber,
\end{IEEEeqnarray}
where $\mathcal{O}^c_{s}$ is the complement of the outage event.

Since $R_s$ and $R_c$ are fixed for all channel and side information states, they are chosen to minimize the expected distortion. Thus, we have
\begin{IEEEeqnarray}{rCl}\label{eq:OptEDBin}
ED_{s}^*\triangleq\min_{R_c,R_s} ED_{s}(R_s,R_c).
\end{IEEEeqnarray}
When the side information has a continuous quasiconcave gain
distribution, we can have a closed-form expression for the optimal
source coding rate $R_s$, as given in the next lemma.
\begin{lemma}\label{lem:OptLay}
For a given $R_c$, if $p_{\Gamma}(\gamma)$ is continuous and quasiconcave, $ED_{s}(R_s,R_c)$ is minimized
by setting $R_s=\frac{1}{2}\log(1+(1+\bar{\gamma} )(2^{2R_c}-1))-R_c+\epsilon$ where
$\bar{\gamma} $ is the solution to (\ref{eq:OptGamma}).
\end{lemma}
\begin{proof}
Once $R_c$ is fixed, it
follows from the results in Section \ref{sec:SourceCoding} that
$ED_{s}(R_s,R_c)$ is minimized by compressing the source to a
single layer targeted for side information state $\bar{\gamma} $,
i.e.,
$R_c=I(S;W|T\!=\!\bar{\gamma}S+Z)\!=\!\frac{1}{2}\log\left(1+(2^{2(R_s+R_c-\epsilon)}-1)/(1+\bar{\gamma})\right)$,
from where $R_s$ is obtained.
\end{proof}


When the side information fading distribution is such that $\bar{\gamma} =0$, then, from Lemma \ref{lem:OptLay}, the optimal source coding rate is $R_s=0$, i.e., the minimum expected distortion is achieved by ignoring the decoder side information in the encoding process.
\begin{corollary}
 If $\bar{\gamma} =0$, the optimal SSCC does not utilize binning, that is, $R_s^*=0$.
\end{corollary}

Note that we have  considered only a single layer source coding scheme since for continuous quasiconcave $p_{\Gamma}(\gamma)$, the optimal source code uses a single source code layer. However, in the case of a discrete side information gain distribution, the optimal source code employs multiple source layers, one layer targeting each side information state \cite{ng2007minimum}. 
\subsection{Joint Decoding Scheme (JDS)}\label{ssec:Joint}
Here, we consider a source-channel coding scheme that does not
involve any explicit binning at the encoder and uses joint decoding
to reduce the outage probability. This coding scheme is introduced in
\cite{tuncel2006slepian} in the context of broadcasting a common
source to multiple receivers with different side information
qualities, and it is shown to be optimal in the case of lossless
broadcasting over a static channel. The success of the decoding
process depends on the joint quality of the channel and
the side information states.

At the encoder, a codebook of $2^{nR_j}$ length-$n$ quantization
codewords $W^n(i)$, $i=1,...,2^{nR_j}$, are generated through a `test
channel' $W=S+Q$, where $Q\sim\mathcal{N}(0,\sigma_Q^2)$ is
independent of $S$. The quantization noise variance is chosen such
that $R_j=I(S;W)+\epsilon$, for an arbitrarily small $\epsilon>0$. Then, an independent Gaussian codebook of size
$2^{nR_j}$ is generated with length-$n$ codewords $X^n(i)$ with
$X\sim \mathcal{N}(0,1)$. Given a source outcome $S^n$, the
transmitter finds the quantization codeword $W^n(i)$ jointly typical
with the source outcome and transmits the corresponding channel
codeword $X^n(i)$ over the channel. At reception, the decoder looks for an index $i$ for which both
$(x^n(i),Y^n)$ and $(T^n,w^n(i))$ are jointly typical. Then the
outage event is given by
\begin{IEEEeqnarray}{rCl}\label{eq:OutageSets}
\mathcal{O}_j&\triangleq&\{(h, \gamma):I(S;W|T)\geq
I(X;Y)\},
\end{IEEEeqnarray}
where $I(S;W|T)=\frac{1}{2}\log\left(1+(2^{2(R_j-\epsilon)}-1)/(\gamma+1)\right)$ and  $I(X;Y)=\frac{1}{2}\log(1+h)$.

If decoding is successful, 
$S^n$ is estimated using both the quantization codeword and the side information sequence, while if an outage occurs, $S^n$ is
reconstructed using only the side information sequence. The expected distortion for the JDS scheme is found as
\begin{IEEEeqnarray}{rCl}\label{eq:Dj}
ED_{j}(R_{j})&=&\mathrm{E}_{\mathcal{O}^c_j}[D_d(R_j,\Gamma)]+\mathrm{E}_{\mathcal{O}_j}[D_d(0,\Gamma)].
\end{IEEEeqnarray}
Similarly to (\ref{eq:OptEDBin}), the expected distortion can be optimized over $R_j$ to obtain the minimum expected distortion achieved by JDS, that is, $ED^*_j\triangleq \min_{R_j}ED_j(R_j)$.

In SSCC, the quantization codeword is successfully decoded only if both the channel and source codewords are successfully decoded. On the other hand,
JDS decodes the quantized codeword exploiting the joint quality of
 the channel and side information. Hence, a bad channel realization can be
compensated with a sufficiently good side information realization, or vice versa, reducing the outage probability. Indeed, the minimum expected distortion of JDS is always lower than that of SSCC, as stated in the next lemma.

\begin{lemma} \label{lemm:EDsbgeqEDx}
For any given $p_{{H}}(h)$ and $p_{\Gamma}(\gamma)$, JDS outperforms SSCC at any SNR, i.e., $ED_{s}^*\geq ED^*_j$.
\end{lemma}

\begin{proof}
Consider the SSCC scheme as in Section \ref{ssec:SeparateBinning} with rates $R_c$ and $R_s$. We will show that the JDS scheme with rate $R_j=R_s+R_c$ achieves a lower expected distortion, i.e., 
$ED_{s}(R_c,R_s)\geq ED_j(R_c+R_s)$. If both schemes are in outage, or if the quantization codeword is decoded successfully in both, they achieve the same distortion. Thus, to prove our claim, it will suffice to show that $\mathcal{O}_{s}\supseteq \mathcal{O}_j$.

 Let $(h,\gamma)$ be such that $R_c\geq I(U;V)=\frac{1}{2}\log(1+h)$, i.e., SSCC is in outage. Note that for  given $(h,\gamma)$, $R_s$ and $R_c$, $I(U;V)$
and $I(W;X|Y)$ have the same values for both schemes. However,
 if $I(W;X|Y)< I(U;V)$, JDS is able to decode the quantization codeword
successfully while SSCC would still be in outage. This condition is satisfied whenever $\frac{1}{2}\log\left(1+\frac{2^{2(R_j-\epsilon)}-1}{\gamma+1}\right)< \frac{1}{2}\log
(1+h)$, or equivalently $\gamma> \frac{2^{2(R_j-\epsilon)}-1}{h}-1$. If this condition does not hold, both schemes are in outage and have the same performance. Conversely, if JDS is in outage, i.e., $I(W;X|Y)\geq I(U;V)$, then SSCC is also in outage since either $R_c\geq I(U;V)$ or $R_c\leq I(U;V)\leq I(W;X|Y)$ holds. Therefore, we have $\mathcal{O}_{s}\supseteq\mathcal{O}_j$, which implies $ED_{s}(R_c,R_b)\geq ED_j(R_c+R_b)$ and $ED_{s}^*\geq ED^*_j$. This completes the proof.
\end{proof}

\subsection{Superposed Hybrid Digital-Analog Transmission (S-HDA)}\label{ssec:S-HDA}
Next, we consider hybrid digital-analog transmission, in which the channel input is generated by a symbol-by-symbol mapping of the observed source sequence and its compression codeword, that is, an analog and a digital signal, respectively. A general HDA scheme is studied in \cite{Minero2015:Hybrid} in the absence of side information for static channels. The necessary conditions on the achievable distortion are derived based on auxiliary random variables, by accounting for the correlation in the indexes of source and channel codebooks. This general hybrid scheme is detailed next.  

Fix a conditional distribution $p(w|s)$, an encoding function $x(s,w)$ and a reconstruction function $\hat{s}(w,y,t)$. At the encoder, generate a codebook of $2^{nR_h}$ length-$n$ codewords $W^n(m)$, $m=1,...,2^{nR_h}$ with i.i.d. components following $p(w)$. The transmitter finds $W^n(m)$ jointly typical with $S^n$ and maps $(S^n,W^n(m))$ symbol-by-symbol to the channel input sequence $X^n$ with the encoding function $x(s,w)$, i.e., $x_i=x(s_i,w_i(m))$, $i=1,...,n$. Upon receiving $Y^n$, the decoder looks for the codeword $W^n$ that is jointly typical with $Y^n$ and the side information $T^n$, and reconstructs $\hat{S}^n$ by mapping symbol-by-symbol the decoded codeword $W^n(\hat{m})$, the analog channel output $Y^n$ and the side information $T^n$  using the reconstruction function $\hat{s}(w,y,t)$. In our setup, the side information, i.e., $T^n=\gamma_c S^n+Z^n$, can be modeled as a channel output correponding to input $S^n$. Then, it follows from \cite{Minero2015:Hybrid} that a distortion $D_h$ is achievable if
\begin{IEEEeqnarray}{rCl}\label{eq:HybDecCond}
I(S;W)<I(W;YT).
\end{IEEEeqnarray} 
holds for some conditional distribution $p(x|s)$, an encoding function $x(s,w)$ and a reconstruction function $\hat{s}(w,y,t)$ such that $\mathrm{E}[(S-\hat{S})^2]\leq D_h$.

In general, it is complicated to characterize the optimal $W$ and channel mappings $x(s,w)$ minimizing the distortion. Here, we propose a particular construction for the time-varying setup, which we denote by superposed hybrid digital-analog transmission (S-HDA), in which source sequence is quantized, and the quantization error is superposed on the source sequence. The power is allocated among the two layers. The uncoded component causes an interference correlated with the source sequence, and acts as side information at the decoder. On the contrary, if an outage occurs and the quantization codeword is not successfully decoded, the analog component provides additional robustness since the channel now contains a noisy uncoded version of the source sequence useful for the reconstruction.
We consider $W^n$, generated using a test channel $W\triangleq \eta S+Q$, where $Q\sim \mathcal{N}(0,1)$ is independent of $S$ and a channel input mapping $x(s,w)$ such that  
\begin{align}\label{eq:HDAChaninput}
X^n=\sqrt{P_d}(W^n-\eta S^n)+\sqrt{P_{a}}S^n,
\end{align}
where $P_d=1-P_a$ is the power allocated to the digital channel input and $P_a\in[0,1]$ is the power allocated to the uncoded layer. We set 
$R_h=I(S;W)+\epsilon$, i.e.,  $\eta^2=P_d(2^{2R_h-\epsilon}-1)$.

An outage will be declared whenever condition (\ref{eq:HybDecCond}) does not hold due to the randomness of the
channel and side information. Hence, the outage
event is defined by
\begin{IEEEeqnarray}{rCl}\label{eq:OutS-HDA}
\mathcal{O}_h\triangleq \{(h,\gamma):I(W;S)\geq I(W;Y,T)\},
\end{IEEEeqnarray}
and is given by
\begin{IEEEeqnarray}{rCl}\label{eq:Outagehybrid}
\mathcal{O}_h\!\triangleq\! \{(h,\gamma)\!:\!P_d h(1+P_d\gamma)\leq
P_d(h(\sqrt{P_a}-\eta)^2)+\eta^2\}.\end{IEEEeqnarray}

If $W^n$ is
successfully decoded, each $S_i$ is reconstructed using an MMSE
estimator with all the information available at the decoder,
$\hat{S}_i=\mathrm{E}[S_i|W_i,Y_i,T_i]$. The corresponding achievable distortion is given by
\begin{IEEEeqnarray}{rCl}\label{eq:DistS-HDA2}
D_h(P_d,\eta)=\frac{P_d}{\eta^2+P_d
\left(1+\gamma+h\left(\sqrt{P_a}-\eta\right)^2\right)}.
\end{IEEEeqnarray}
If an outage occurs, the
receiver estimates $X^n$ from $T^n$ and $Y^n$ with an MMSE estimator,
$\hat{X}_i=\mathrm{E}[X_i|V_i,Y_i]$. The achieved distortion is found to be
\begin{IEEEeqnarray}{rCl}
D_{h}^{out}(P_d,\eta)\triangleq\left(1+\frac{h P_a}{1+h P_d}+\gamma\right)^{-1}.
 \end{IEEEeqnarray}
Finally, the expected distortion for S-HDA is given by
\begin{IEEEeqnarray}{rCl}\label{eq:DistS-HDA}
ED_{shda}(P_d,\eta)\triangleq \mathrm{E}_{\mathcal{O}^c_h}[D_h(P_d,\eta)]+\mathrm{E}_{\mathcal{O}_h}[D_h^{out}(P_d,\eta)].
\end{IEEEeqnarray}
Optimizing over $P_d$ and $\eta$, we obtain $ED^*_{shda}\triangleq \min_{P_d,\eta}ED_{shda}(P_d,\eta)$. Note that uncoded transmission can be recovered from $ED_{shda}(P_d,\eta)$ with $P_d=0$. The hybrid digital analog (HDA) scheme of \cite{wilson2010joint} can be recovered by letting $P_a=0$. We define the minimum expected distortion achievable with HDA as $ED^*_{hda}\triangleq \min_{\eta}ED_{shda}(1,\eta)$.

\begin{remark}
We note that JDS can also be derived from the general hybrid scheme by letting $W=(W',X')$, where $W'$ is the quantization codeword and $X'$ is the channel input in Section \ref{ssec:Joint}, and using the mapping $x(s,w)=x'$, i.e., $X'$ is used as the channel input. Note that while both can be derived from the general hybrid scheme, JDS and S-HDA are operationally different and neither of them is a special  case of the other.
\end{remark}


\section{Optimality of Uncoded Transmission}\label{sec:OptimalityUncoded}

In addition to separate source and channel coding, uncoded transmission is well known to achieve the minimum distortion in point-to-point static Gaussian channels  \cite{Goblick:IT:65}, \cite{Gastpar:IT:03}.
In the presence of channel fading, separate source and channel coding becomes suboptimal while uncoded transmission still achieves the optimal performance, due to its robustness to channel variations.
  Scenarios in which uncoded transmission achieves the optimal performance have received a lot of attention in the literature, such as the transmission of noisy observations of a Gaussian source over Gaussian multiple access channels (MACs)\cite{Gaspart2008Uncoded} and the transmission of correlated Gaussian sources over a Gaussian MAC, in which case the uncoded transmission is shown to be optimal below a certain SNR threshold \cite{lapidoth2010sending}.

However, even in a point-to-point Gaussian channel, in the presence of static side information at the decoder, uncoded transmission becomes suboptimal. In this case, concatenating a Wyner-Ziv source code with a capacity achieving channel code \cite{Shamai:IT:98}, or joint source-channel coding through the  HDA scheme of \cite{wilson2010joint} is required to achieve the optimal distortion.
We show below that, in the block fading scenario studied here, when the side information state $\Gamma$ follows a continuous and quasiconcave distribution for which $\bar{\gamma} =0$ is the solution to equation (\ref{eq:OptGamma}), uncoded transmission meets the lower bound $ED^*_{pi}$ in (\ref{eq:GenlowerBoundSISO}) for any  channel gain distribution $p_{H}(h)$. Hence, despite the presence of correlated side information, uncoded transmission achieves the optimal performance, while both the separate source and channel coding and HDA schemes are suboptimal. The optimality of uncoded transmission follows since, when $\bar{\gamma}=0$, the side information is ignored in the encoding even for the partially informed lower bound. 

\begin{theorem}\label{th:OptUnc}
Let $p_{{H}}(h)$ be an arbitrary pdf while $p_{\Gamma}(\gamma)$
is a continuous and quasiconcave function satisfying equation (\ref{eq:OptGamma}) for $\bar{\gamma} =0$. Then, the minimum expected distortion $ED^*$ is achieved by uncoded transmission.
\end{theorem}

\begin{proof}
For any pdf satisfying (\ref{eq:OptGamma}) with $\bar{\gamma} =0$, the partially informed encoder lower bound is given by
\begin{IEEEeqnarray}{rCl}
ED_{pi}^*&=&\left.\mathrm{E}_{{H}}\left[ED^*_Q\left(\frac{1}{2}\log(1+{H})\right)\right]\right|_{\bar{\gamma} =0}\nonumber\\
&\stackrel{(a)}{=}&\int_{h}\int_{0}^{\infty}\frac{p_{{H}}(h)p_{\Gamma}(\gamma )}{2^{\log(1+h)}+\gamma }d\gamma dh\nonumber\\
&=&\iint_{h,\gamma }\frac{p_{{H}}(h)p_{\Gamma}(\gamma )}{1+h+\gamma }d\gamma dh\nonumber\\
&=&ED_u,\nonumber
\end{IEEEeqnarray}
where $(a)$ is obtained by substituting $\bar{\gamma} =0$ in (\ref{eq:SingleLayerCode}). This completes the proof.
\end{proof}

The class of continuous quasiconcave functions for which any non-empty super-level set of $f_{\Gamma}(\gamma )$ begins at $\gamma =0$ satisfies $\bar{\gamma} =0$. It is not hard to see that the class of continuous monotonically decreasing functions in $\gamma \geq0$ satisfies this condition.
\begin{proposition}\label{prop:MonotonicallyDecresaing}
Let $p_{\Gamma}(\gamma )$ be a continuous monotonically decreasing function for $\gamma >0$. Then, (\ref{eq:OptGamma}) holds for $\bar{\gamma} =0$; and hence, uncoded transmission achieves the optimal performance.
\end{proposition}

\begin{proof}
 By definition $\bar{\gamma} $ is given by the left endpoint of the super-level set induced by $\alpha^*$. For any monotonically decreasing function $p_{\Gamma}(\gamma )$, the left endpoint of the super-level set $\{\gamma: p_{\Gamma}(\gamma )\geq\alpha\}$ corresponds to $\gamma =0$, and as a consequence, we have $\bar{\gamma} =0$ for any value of $\alpha^*$.
\end{proof}

\section{Finite SNR Results}\label{sec:FiniteSNRResults}

In the previous section we have seen the optimality of uncoded transmission when the side information fading state follows a continuous quasiconcave pdf for which $\bar{\gamma} =0$. The exponential distribution, and the more general family of gamma distributions with shape parameter $L\leq1$ are continuous monotonically decreasing distributions, and hence, the uncoded transmission is optimal when the side information gain $\Gamma$ follows one of these distributions. Gamma distributed fading gains appear, for example, when the channel state follows a Nakagami distribution. The gamma distribution with shape parameter $L$ and scale parameter $\theta$, $\Gamma\sim\Upsilon(L,\theta)$, is given as
\begin{IEEEeqnarray}{rCl}\label{eq:Gammapdf}
p_{\Gamma}(\gamma)=\frac{1}{\theta^L}\frac{1}{\Psi(L)}\gamma^{L-1}e^{-\frac{\gamma}{\theta}},\; \text{for }\gamma\!\geq0,\text{and }L,\theta>0,
\end{IEEEeqnarray}
where $\Psi(z)\triangleq \int_{0}^{\infty}t^{z-1}e^{-t}dt$ is the gamma function. The variance of $\Gamma$ is $\sigma^2_{\Gamma}=L\theta^2$ and its mean is $\mathrm{E}[\Gamma]=L\theta$. When $L\leq 1$, it is easy to check that $p_{\Gamma}(\gamma)$ is continuous monotonically decreasing, while it is continuous quasiconcave for $L>1$. Note that when $L=1$, the gamma distribution reduces to the exponential distribution.

Parameter $L$ models the \emph{side information diversity} since a time-varying side information sequence $Y^m$, with state distribution $p_{\Gamma}(\gamma)$, provides the equivalent information (in the sense of sufficient statistics) provided by $L$  independent side information sequences each with i.i.d. Rayleigh block-fading gains. We note that despite the term ``diversity", the side information diversity comes from uncoded noisy versions of the source sequence; hence, the gain it provides is limited compared to the channel diversity which can be better exploited through coding.

To illustrate the performance of the achievable schemes and compare them with the lower bounds, we consider Nakagami fading channel and side information distributions.
We consider normalized
channel and side information gains ${H}_c=\sqrt{\rho} {H}_{c0}$
and $\Gamma_c=\sqrt{\rho} \Gamma_{c0}$, such that
\begin{IEEEeqnarray}{rCl}
Y^n=\sqrt{\rho}{H}_{c0}X^n+N^n,\quad
T^n=\sqrt{\rho}\Gamma_{c0}S^n+Z^n,\nonumber
\end{IEEEeqnarray}
where ${H}_{c0}$ and $\Gamma_{c0}$ satisfy $\mathrm{E}[{H}_{c0}^2]=\mathrm{E}[\Gamma_{c0}^2]=1$. Basically, ${H}_{c0}$ and $\Gamma_{c0}$ capture the randomness in the channels while $\rho$ is the average SNR. We define the associated instantaneous gains $H_{0}\triangleq {H}_{c0}^2$ and $\Gamma_{0}\triangleq \Gamma_{c0}^2$.
\begin{figure}[t!]
\centering
\includegraphics[width=0.7\textwidth]{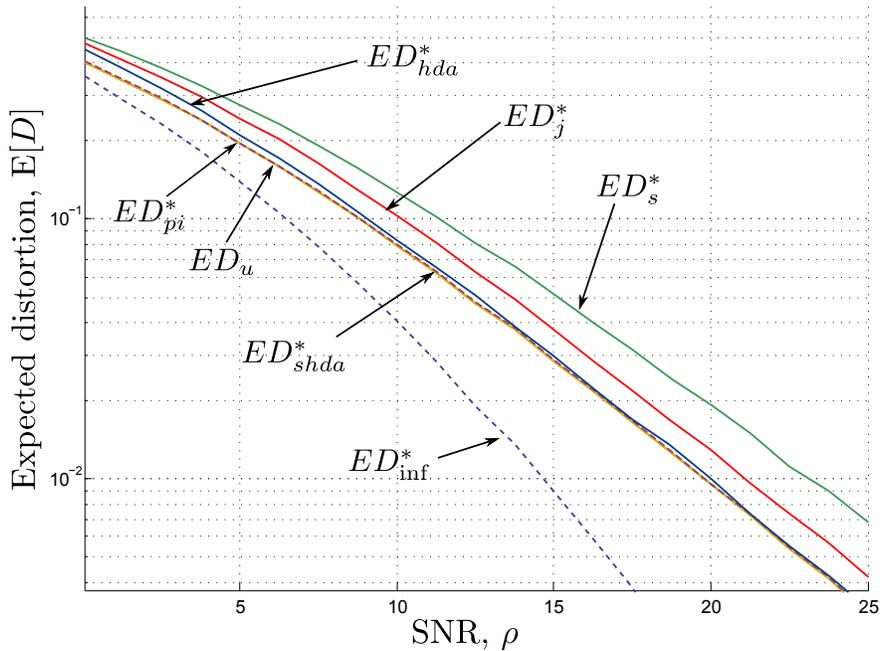}
 \caption{Upper and lower bounds on the expected distortion versus the channel SNR $(\rho)$ for Rayleigh fading channel and side information gain distributions, i.e., $L_s=L_c=1$,  with $\rho=\mathrm{E}[H^2_c]=\mathrm{E}[\Gamma^2_c]$.} \label{fig:Distortion}
\end{figure}
We assume that the channel gain $H_0$ has a gamma distribution with scale parameter $L_c>0$ and $\theta_c= L_c^{-1}$, i.e., $H_{0}\sim \Upsilon(L_c,L_c^{-1})$, and similarly, the side information gain follows a gamma distribution with $L_s>0$ and $\theta_s= L_s^{-1}$, i.e., $\Gamma_{0}\sim \Upsilon(L_s,L_s^{-1})$. We have fixed the value of $\theta_c$ and $\theta_s$, such that $\mathrm{E}[H_{c0}^2]=\mathrm{E}[H_{0}]=1$ and $\mathrm{E}[\Gamma_{c0}^2]=\mathrm{E}[\Gamma_{0}]=1$, and both channels have the same average SNR $\rho$ for any $L_c$ and $L_s$. Note that the variance of $\Gamma$ is $\sigma^2_{\Gamma}=L_s\theta^2=1/L_s$. Thus, the side information gain $\Gamma$ becomes more deterministic as $L_s$ increases, and similarly, for $L_c$ and $H$.

First we consider the case with $L_s=L_c=1$, i.e., both the channel and the side information gains are Rayleigh distributed. In Figure \ref{fig:Distortion} we plot the expected distortion with respect to the channel SNR. As shown in Theorem \ref{th:OptUnc}, uncoded transmission achieves the partially informed encoder lower bound
$ED^*_{pi}$.  The minimum expected distortion is given by
\begin{IEEEeqnarray}{rCl}\label{eq:RayleighDistortion}
\!\!ED^*\!\!=\!ED_u\!=\!\!\int_{h_{0}}\!\frac{1}{\rho}e^{\frac{1+\rho h_{0}}{\rho}}\!E_{1}\left(\frac{1+\rho h_{0}}{\rho}\right)p_{{H}_0}(h_{0})d h_{0}.
\end{IEEEeqnarray}
 We see from the figure that the informed encoder lower bound is significantly loose, especially at high SNR. This gap between the two lower bounds also illustrates the potential performance improvement that will be achieved by increasing the feedback resources. If both channel and side information states can be fed back to the encoder, instead of only CSI feedback, a significant improvement can be achieved. In relation to this observation, a problem that requires further research is the allocation of feedback resources between channel and side information states when a limited feedback channel is available from the decoder to the encoder.

S-HDA ($ED^*_{shda}$) also achieves the optimal performance by allocating all available power to the analog component, reducing it to uncoded transmission. Note that while the HDA scheme of \cite{wilson2010joint} cannot reach $ED^*$ in the low SNR regime, its performance gets very close to $ED^*$ at high SNR values.

The expected distortion achievable by SSCC is minimized without any binning, since we have $\bar{\gamma} =0$ for Rayleigh fading side information. Hence,  $R^*_s=0$ from Lemma \ref{lem:OptLay}. It is interesting to observe that for Rayleigh fading side information states, the uncertainty in the side
information renders it useless for the encoder, and the side information is ignored to avoid outages in source decoding, and it is used only in the estimation step. As will be seen next, this is not the case when the side information fading has a different distribution.

 We also observe in Fig. \ref{fig:Distortion} that JDS ($ED_{j}^*$) outperforms SSCC by exploiting the joint quality of the channel and  side information, as claimed by Lemma \ref{lemm:EDsbgeqEDx}. We also see that JDS cannot achieve the optimal performance in this setting.

Observations above, including the optimality of uncoded transmission, hold for any $L_c$ value as long as $L_s\leq 1$. This follows from Proposition \ref{prop:MonotonicallyDecresaing} since $p_{\Gamma}(\gamma)$ is monotonically decreasing if $L_s\leq 1$. However, this optimality does not hold in general. Next, it will be shown that  uncoded transmission is suboptimal for a wide variety of channel distributions, while S-HDA performs very close to the partially informed encoder lower bound.

\begin{figure}[t!]
\centering
\includegraphics[width=0.7\textwidth]{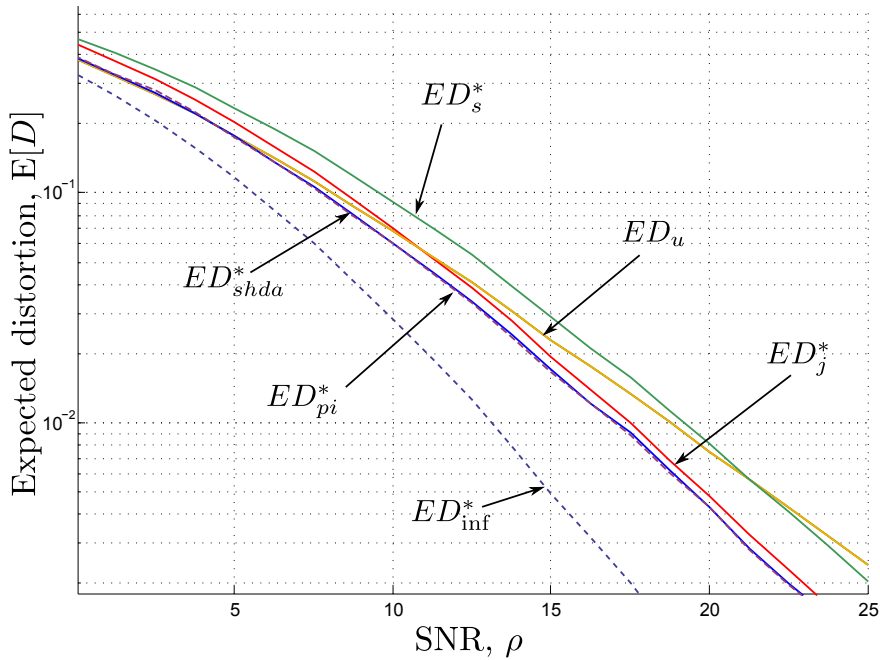}
 \caption{Lower and upper bounds on the expected
distortion versus the channel SNR for $L_s=2$ and $L_c=1$ with
$\rho=\mathrm{E}[H_c^2]=\mathrm{E}[\Gamma_c^2]$.} \label{fig:Diversity2}
\end{figure}
We consider the case with $L_s=2$ and $L_c=1$ in Fig. \ref{fig:Diversity2}. We can see
that S-HDA achieves the lowest expected distortion among the proposed schemes and performs very close to the lower bound at all SNR values, while uncoded transmission is suboptimal. Although  the performance of uncoded transmission is very close to $ED^*_{pi}$ in the low SNR regime, the gap between the two increases with SNR. In addition, both SSCC and JDS surpass the performance of uncoded transmission as the SNR increases.

Finally, in Fig \ref{fig:Diversity30}, we consider the scenario with $L_c=0.5$ and $L_s=1.5$. Contrary to the previous scenarios, in this setup JDS outperforms S-HDA for SNR values greater than $\text{SNR}\backsimeq37\text{dB}$. As the SNR increases, JDS performs close to the partially informed lower bound, while S-HDA performance is further from the lower bound. Similarly to the previous scenarios, we observe that uncoded transmission performs close to the lower bound for low SNR values.

Additional numerical simulations indicate that, as the side information diversity, $L_s$,  increases, the gap at any SNR between the informed encoder lower bound and the partially informed encoder lower bound reduces.
The two bounds converge since for the studied setup $\sigma^2_{\Gamma_{c0}}=L_s^{-1}$, and as $L_s$ increases, the variance decreases, and therefore, the level of
uncertainty in the side information gain state diminishes.  In fact, the two bounds can be shown to converge at any SNR value and for any arbitrary side information gain distribution whose variance decreases with some parameter, namely $L_s$, as given in the next lemma.

\begin{lemma}\label{lem:ConvergentBounds}
Let ${H}$ be arbitrarily distributed and have a finite mean, i.e., $\mathrm{E}_{{H}}[{H}]<\infty$. Let $\left(\Gamma_{L}\right)_{L\geq0}$ be a sequence of side information gain random variables such that, for every $L$, $\Gamma_{L}$ follows an arbitrary distribution with variance $\sigma^2_{L}$, where $\sigma^2_{L}\rightarrow0$ as  $L\rightarrow\infty$. Then, as  $L\rightarrow\infty$, the partially informed encoder lower bound converges to the informed encoder lower bound, i.e., the following limit holds:
\begin{IEEEeqnarray}{rCl}
\lim_{L\rightarrow \infty}(ED_{\mathrm{inf}}-ED^*_{\mathrm{pi}})=0.
\end{IEEEeqnarray}
\end{lemma}
\begin{proof}
See Appendix  \ref{app:ConvergentBounds}.
\end{proof}
Although the side information available at the decoder becomes more deterministic with increasing $L_s$, the channel is still block-fading. Only S-HDA performs close to the informed encoder lower bound, i.e., the optimal performance when the current channel and side information states are known. On the contrary, the rest of the studied schemes cannot fully exploit the determinism in the side information fading gain for $L_c\geq1$, while it seems that for $L_c<1$ JDS is the scheme achieving the lowest expected distortion.

\begin{figure}[t!]
\centering
\includegraphics[width=0.7\textwidth]{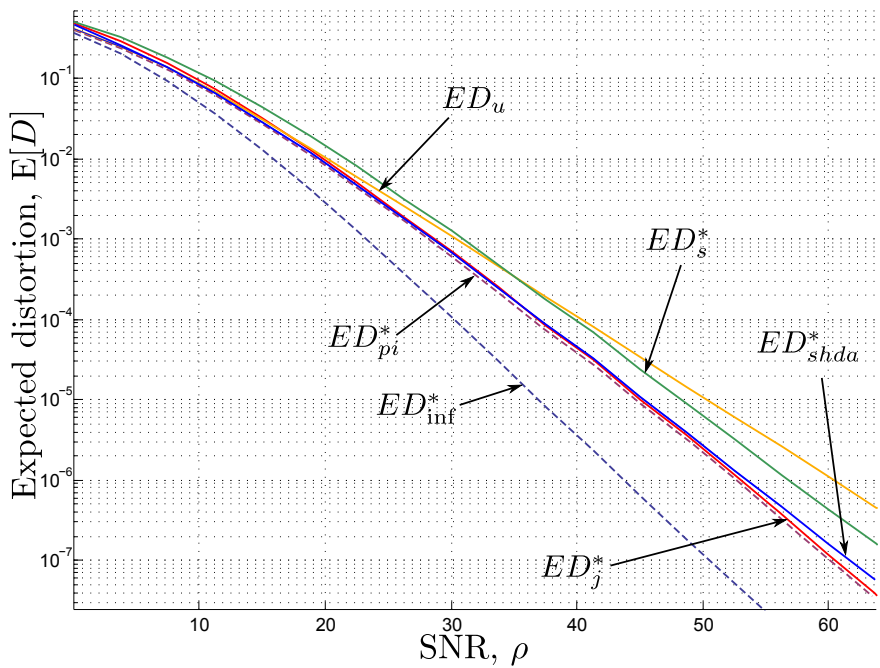}
 \caption{Lower and upper bounds on the expected
distortion versus the channel SNR for $L_s=1.5$ and $L_c=0.5$ with
$\rho=\mathrm{E}[H_c^2]=\mathrm{E}[\Gamma_c^2]$.}\label{fig:Diversity30}
\end{figure}

\section{High SNR Analysis}\label{sec:HighSnrAnalysis}
In the previous section we have seen the optimality of uncoded transmission in certain settings in which the proposed digital schemes are suboptimal. On the other hand, our numerical results have shown that the S-HDA scheme has a good performance for a wide variety of channel distributions while the optimality of uncoded transmission is very sensitive to the distribution of the side information. We have also observed that JDS outperforms S-HDA in certain regimes.
 Although we have characterized the optimal expected distortion in closed-form for the Rayleigh fading scenario in (\ref{eq:RayleighDistortion}), a closed-form expression of the optimal expected distortion for general channel and side information distributions is elusive. Instead, we focus on the high SNR regime, and study the exponential decay rate of the expected distortion with increasing SNR, defined as the \emph{distortion exponent}, and denoted by $\Delta$ \cite{Laneman2005SourceChannelDiv}. We have,
\begin{IEEEeqnarray}{rCl}
\Delta\triangleq -\lim_{\rho\rightarrow\infty}\frac{\log \mathrm{E}[D]}{\log \rho}.
\end{IEEEeqnarray}

In this section, we study the distortion exponent for the model considered in Section \ref{sec:FiniteSNRResults}, i.e., a Nakagami fading channel and side information gains, i.e.,  ${H}_{0}\sim\Upsilon(L_c,L_c^{-1})$ and $\Gamma_{0}\sim\Upsilon(L_s,L_s^{-1})$.
 We are interested in characterizing the maximum distortion exponent over all encoder and decoder pairs, denoted by $\Delta^*(L_s,L_c)$.

 We first provide an upper bound on the distortion exponent by studying the partially informed encoder lower bound on the expected distortion in (\ref{eq:GenlowerBoundSISO}). In determining the high SNR behavior of the partially informed encoder lower bound, it is challenging to characterize the optimal SNR exponent for the target side information state $\bar{\gamma}$ in (\ref{eq:OptGamma}) for different channel states. Hence, we further bound the expected distortion by considering the ergodic channel capacity as the channel rate.

\begin{lemma}\label{lem:DistExpPartially}
 The optimal distortion exponent is upper bounded by the exponent of the partially informed encoder lower bound calculated at the ergodic channel capacity, given by \begin{IEEEeqnarray}{rCl}
\Delta_{pe}(L_s,L_c)=1+\left(1-\frac{1}{L_s}\right)^+.
\end{IEEEeqnarray}
\end{lemma}
\begin{proof}
See Appendix \ref{app:PartiallyInformedGeneral}.
\end{proof}

\begin{figure}[t!]
\centering
\includegraphics[width=0.6\textwidth]{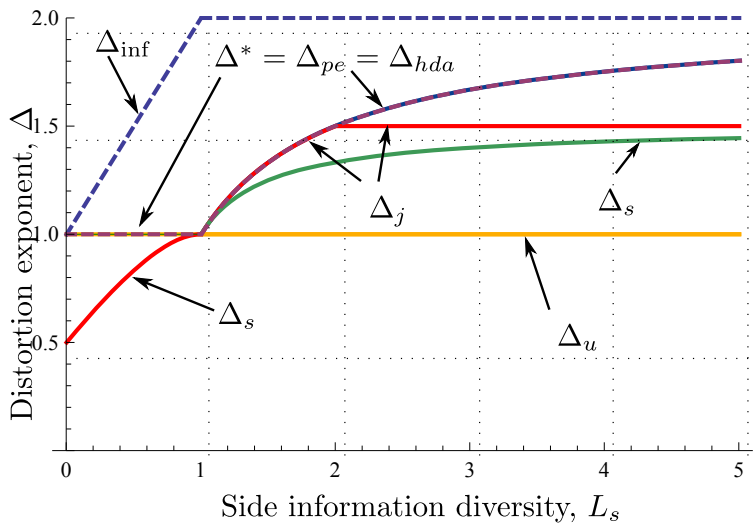}
 \caption{Distortion exponent upper and lower bounds for Nakagami fading channel and side information with $L_c=1$, as a function of $L_s$.}
\label{fig:DistExp}
\end{figure}

\begin{figure}[t!]
\centering
\includegraphics[width=0.6\textwidth]{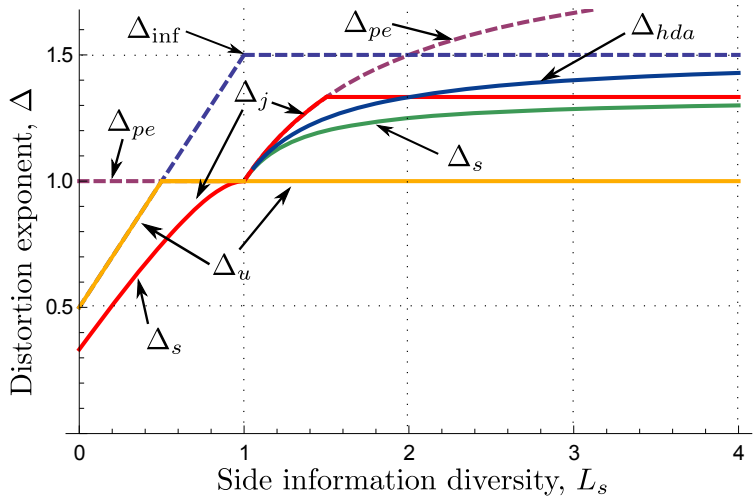}
 \caption{Distortion exponent upper and lower bounds for Nakagami fading channel and side information with $L_c=0.5$, as a function of $L_s$.}
\label{fig:DistExpL205}
\end{figure}

We will see that $\Delta_{pe}(L_s,L_c)$ is tight only for $L_c\geq1$, and the ergodic channel relaxation is loose for $L_c<1$. In order to tighten the bound in these regimes, we consider the distortion exponent of the informed encoder upper bound proposed in Section \ref{sec:UpperAndLowerBounds}, which can be proven similarly to Lemma \ref{lem:DistExpPartially}.

\begin{lemma}
The distortion exponent is upper bounded by the exponent of the informed encoder lower bound, given by
\begin{IEEEeqnarray}{rCl}
\Delta_{\mathrm{inf}}(L_s,L_c)=\min\{L_c,1\}+\min\{L_s,1\}.
\end{IEEEeqnarray}
\end{lemma}


While for $L_c\geq1$, $\Delta_{pe}(L_s,L_c)$ is always tighter than $\Delta_{\mathrm{inf}}(L_s,L_c)$, for $L_c<1$ we have $\Delta_{pe}(L_s,L_c)\geq \Delta_{\mathrm{inf}}(L_s,L_c)$ if $L_s\geq\frac{1}{1-L_c}$. In the next proposition, we combine the two upper bounds into a single upper bound on the distortion exponent.
\begin{theorem}
For a Nakagami fading channel with ${H}_{0}\sim\Upsilon(L_c,L_c^{-1})$, and a Nakagami fading side information with $\Gamma_{0}\sim \Upsilon(L_s,L_s^{-1})$, the optimal distortion exponent is upper bounded by
\begin{IEEEeqnarray}{rCl}
\min\{\Delta_{pe}(L_s,L_c),\Delta_{\mathrm{inf}}(L_s,L_c)\}=\begin{cases}
\min\{1,L_s+L_c\}&\text{if }L_s\leq1,\\
2-\frac{1}{L_s}& \text{if }1<L_s\leq \frac{1}{(1-L_c)^+},\\
1+L_c& \text{if }L_s> \frac{1}{(1-L_c)^+}.
\end{cases}
\end{IEEEeqnarray}
\end{theorem}

In Fig. \ref{fig:DistExp} and Fig. \ref{fig:DistExpL205} we plot the distortion exponent upper and lower bounds with respect to the parameter $L_s$ of the Nakagami distribution for $L_c=1$ and $L_c=0.5$, respectively.
Note that for $L_c\geq1$, as $L_s$ increases, the distortion exponent upper bound $\Delta_{pe}(L_s,L_c)$ converges to the informed encoder upper bound. This observation is parallel to  Lemma \ref{lem:ConvergentBounds}. However, this is not the case if $L_c<1$. While Lemma \ref{lem:ConvergentBounds} applies to any channel distribution, the partially informed bound with ergodic channel relaxation is loose in this regime.

Next, we consider the distortion exponent achievable by the transmission schemes proposed in Section \ref{sec:UpperAndLowerBounds}. The proofs can be found in Appendix \ref{app:DiestExpDerivations}.  We note that the distortion exponent achievable by uncoded tranmission is provided without proof and can be derived similarly to the proofs in Appendix \ref{app:DiestExpDerivations}.

\begin{lemma}
The distortion exponent achieved by uncoded transmission is given by
\begin{IEEEeqnarray}{rCl}
\Delta_u(L_s,L_c)=\min\{L_s+L_c,1\}.
\end{IEEEeqnarray}
\end{lemma}
As expected from Theorem \ref{th:OptUnc}, uncoded transmission achieves the optimal distortion exponent for $L_s\leq1$, and it is suboptimal for $L_s>1$.

\begin{lemma}
The distortion exponent achievable by SSCC is given by
\begin{IEEEeqnarray}{rCl}
\Delta_{s}(L_s,L_c)=\begin{cases}
1-\frac{(1-L_s)^2}{L_c+1-L_s}&\text{if }L_s\leq 1,\\
\frac{L_s(2 L_c+1) -L_c-1}{L_s(L_c+1)-1}&\text{if }L_s>1.
\end{cases}
\end{IEEEeqnarray}
\end{lemma}
Note that when $L_s=1$, SSCC achieves the optimal distortion exponent of $1$. 

\begin{lemma}
The distortion exponent achievable by JDS is given by
\begin{IEEEeqnarray}{rCl}
\Delta_j(L_s,L_c)=\begin{cases}
1-\frac{(1-L_s)^2}{L_c+1-L_s}&\text{if } L_s\leq1,\\
2-\frac{1}{L_s}&\text{if }1<L_s\leq 1+L_c,\\
1+\frac{L_c}{L_c+1} & \text{if }L_s> L_c+1.
\end{cases}
\end{IEEEeqnarray}
\end{lemma}

JDS achieves the same distortion exponent as SSCC for $L_s\leq 1$. However, interestingly, for $1\leq L_s\leq 1+L_c$, JDS achieves the optimal distortion exponent and then saturates for $L_s>1+L_c$. Observe that, as $L_s$ increases, the achievable distortion exponent with SSCC converges to the performance of JDS.

\begin{lemma}\label{lem:S-HDAexponent}
The distortion exponent achievable by S-HDA and HDA is given by
\begin{IEEEeqnarray}{rCl}\label{eq:S-HDAexponent}
\Delta_{shda}(L_s,L_c)=\min\{1,L_s+L_c\}+\frac{\min\{1,L_c\}(L_s-1)^+}{L_s-1+\min\{1,L_c\}}.
\end{IEEEeqnarray}
\end{lemma}

Lemma \ref{lem:S-HDAexponent} reveals that the robustness provided by the uncoded layer in S-HDA is not required in the high SNR regime to achieve the optimal distortion exponent, and allocating all the available power to the HDA layer of the S-HDA scheme is sufficient. However, we remark that, in terms of the expected distortion in the low SNR regime pure HDA is not sufficient to achieve a performance close to the lower bound, and the uncoded layer improves the performance in general, as observed in the previous section.

HDA achieves the optimal distortion exponent for $L_c\geq1$ while the rest of the proposed schemes are suboptimal. However, when $L_c<1$, JDS outperforms HDA for $1\leq L_s\leq2$. Nevertheless, as $L_s$ increases, HDA converges to the distortion exponent of the informed encoder lower bound, despite the uncertainty in the channel state.

We can see that in the limit $L_s\rightarrow \infty$, with $0<L_c\leq 1$, we have
\begin{IEEEeqnarray}{rCl}
\Delta^*(\infty,L_c)=\Delta_{\mathrm{inf}}(\infty,L_c)=\Delta_{hda}(\infty,L_c)=1+L_c,\nonumber
\end{IEEEeqnarray}
whereas
\begin{IEEEeqnarray}{rCl}
\Delta_{s}(\infty,L_c)=\Delta_{j}(\infty,L_c)=1+\frac{L_c}{L_c+1}<1+L_c.\nonumber
\end{IEEEeqnarray}
This result suggests that, as the side information fading state becomes more deterministic, the performance of HDA converges to the informed encoder lower bound, while the rest of the schemes perform significantly worse than HDA.

Combining the achievable distortion exponents of the JDS and HDA schemes, we can characterize the optimal distortion exponent $\Delta^*(L_s,L_c)$ in certain regimes, as given next. See Figure \ref{fig:TheoFig} for an illustration of the schemes achieving the optimal distortion exponent.

\begin{theorem}\label{th:OptimalDistortionExp}
Consider a Nakagami fading channel with ${H}_{0}\sim\Upsilon(L_c,L_c^{-1})$ and a Nakagami fading side information with $\Gamma_{0}\sim \Upsilon(L_s,L_s^{-1})$. If $L_c\geq1$, the optimal distortion exponent is achieved by  the HDA scheme, and is given by
\begin{IEEEeqnarray}{rCl}
\Delta^*(L_s,L_c)=1+\left(1-\frac{1}{L_s}\right)^+.
\end{IEEEeqnarray}

If $L_c< 1$, and $L_s\leq 1+L_c$, the optimal distortion exponent is given by
\begin{IEEEeqnarray}{rCl}
\Delta^*(L_s,L_c)=\min\{1,L_s+L_c\}+\left(1-\frac{1}{L_s}\right)^+,
\end{IEEEeqnarray}
and is achieved by uncoded transmission and HDA when $L_s\leq 1$, and by JDS when $1\leq L_s\leq L_c+1$.
\end{theorem}

\begin{figure}[t!]
\centering
\includegraphics[width=0.6\textwidth]{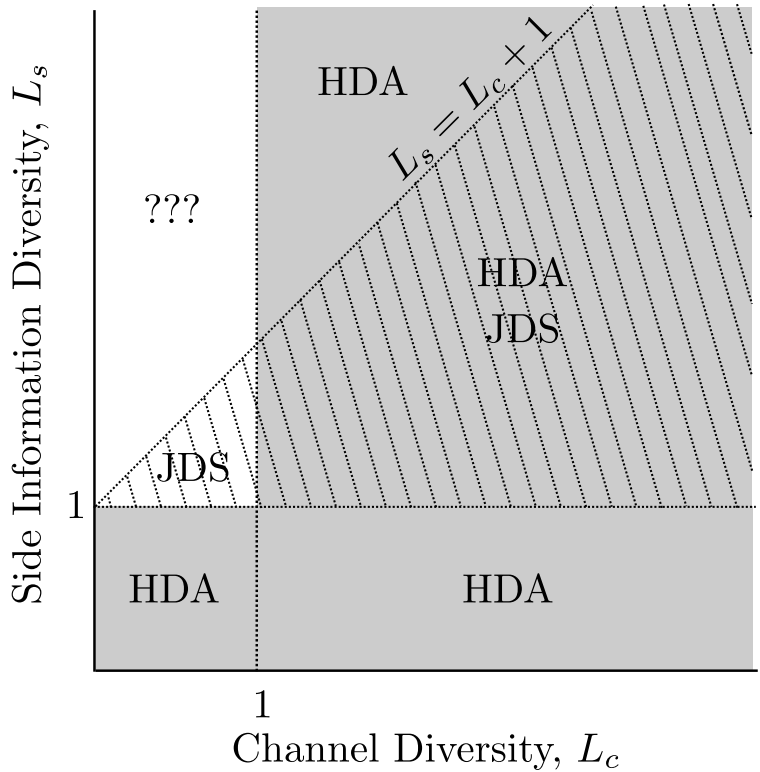}
 \caption{Illustration of the results in Theorem \ref{th:OptimalDistortionExp}. The schemes achieving the optimal distortion are included in each regime.}
\label{fig:TheoFig}
\end{figure}

These analytical results are in line with the numerical analysis carried out in Section \ref{sec:FiniteSNRResults}. For $L_s=L_c=1$, all the schemes achieve the optimal distortion exponent $\Delta^*(1,1)=1$, which is far from the informed encoder upper bound given by $\Delta_{\mathrm{inf}}(1,1)=2$, as observed in Fig. \ref{fig:Distortion}. For $L_s=2$ and $L_c=1$, plotted in Fig. \ref{fig:Diversity2}, the optimal distortion exponent is given by $\Delta^*(2,1)=3/2$, which is achieved by HDA, while uncoded transmission is suboptimal since $\Delta_u(2,1)=1$. In this case JDS also achieves the optimal distortion exponent, while SSCC achieves a lower distortion exponent of $\Delta_{s}(2,1)=4/3$.
Although a similar behavior is observed for higher values of $L_s$, JDS does not achieve the optimal distortion exponent in general.
However, when $L_c=0.5$ and $L_s=1.5$, plotted in Fig. \ref{fig:Diversity30}, JDS achieves the optimal distortion exponent of $\Delta^*(1.5,0.5)=4/3$, while HDA achieves a smaller distortion exponent given by $\Delta_{shda}(1.5,0.5)=5/4$. 

\vspace{-5mm}

\section{Conclusions}\label{sec:Conclusions}

We have studied the joint source-channel coding problem of transmitting a Gaussian source over a delay-limited block-fading channel when block-fading side information is available at the decoder. We have assumed that  the
receiver has full knowledge of the channel and side information states while the transmitter is aware only of their distributions. In the case of a static channel, we have shown the optimality of separate source and channel coding when the side information gain follows a discrete or a continuous quasiconcave distribution.

When both the channel and side information states are block-fading, the optimal performance is not known in general. We have proposed achievable schemes based on uncoded transmission, separate source and channel coding, joint decoding and hybrid digital-analog transmission. We have also derived a lower bound on the expected distortion by providing the encoder with the actual channel state. We call this the partially informed encoder lower bound, since the side information state remains unknown to the encoder. We have shown that this lower bound is tight for a certain class of continuous quasiconcave side information fading distributions, and the optimal performance is achieved by uncoded transmission.  This, to the best of our knowledge, constitutes the first communication scenario in which the uncoded transmission is optimal thanks to the existence of fading, while it would be suboptimal in the static case. We have also proved that joint decoding outperforms separate source and channel coding since the success of decoding at the receiver depends on the joint quality of the channel and side information states, rather than being limited by each of them separately. We have also shown  numerically that the proposed superposed hybrid digital-analog transmission performs very close to the lower bound for a wide range of channel and side-information distributions (in particular, we have considered Gamma distributed channel and side information gains with different shape parameters). However, it has also been observed that no unique transmission scheme outperforms others at all cases.

In the high SNR regime, we have obtained closed-form expressions for the distortion exponent, i.e., the optimal exponential decay rate of the expected distortion in the high SNR regime, of the proposed upper and lower bounds for Nakagami distributed channel and side information states. Aligned with the numerical results in the finite SNR regime, we have shown that the superposed hybrid digital-analog transmission outperforms other schemes in most cases and achieves the optimal distortion exponent for certain values of channel and side information diversity, and joint decoding achieves the optimal distortion exponent for some values of side information diversity when the channel diversity is less than one, in which case hybrid digital-analog transmission is suboptimal.


\appendices

\section{Proof of Theorem \ref{th:separation}}\label{app:separation}


The theorem is first proven when $\Gamma$ has a discrete distribution.
For $\Gamma$ with two states optimality of separation can be obtained as a special case of the model studied in \cite{Steinberg2006HierarchicalJointSourceChannel}. This result can be extended to $M$ receivers (or states). The converse follows by combining the
converses in \cite{Steinberg2006HierarchicalJointSourceChannel} and
\cite[Sec.VII]{Berger1985SideInfMayBeAbsent} for $M$ side
information states, i.e., $Y_{i,1}^n$, $i=1,...,M$, and the application of standard arguments. We obtain the single letter
condition,
\begin{IEEEeqnarray}{rCl}\label{eq:Conv}
\mathrm{C}\triangleq \max_{p(x)}I(X;Y)\geq R_{\mathrm{HB}}(\mathbf{D}),
\end{IEEEeqnarray}
where $R_{\mathrm{HB}}(\mathbf{D})$ is 
Heegard-Berger rate-distortion function for $M$ side information states \cite{Berger1985SideInfMayBeAbsent} and $\mathbf{D}=[D_1,...,D_M]$ with $D_i$ defined as the achievable distortion at the receiver $i$. We note that $ R_{\mathrm{HB}}(\mathbf{D})$
does not depend on the number of receivers but only on the sum
of the mutual information terms, each one corresponding to a receiver with side information $Y_i$, as discussed in
\cite{ng2007minimum}. Hence, the converse applies for countably
many receivers as well.
The achievability follows from the concatenation of the optimal  Heegard-Berger source code \cite[Sec.VII]{Berger1985SideInfMayBeAbsent}, followed by an optimal channel code at rate $R=\mathrm{C}$.

\vspace{-4mm}
\subsection{Separation for Continuous Quasiconcave Distributions}

To prove the optimality of separation when $p_{\Gamma}(\gamma )$ is a continuous quasiconcave distribution, we construct a lower bound on the expected distortion $ED^*_{sta}$ by discretizing the continuum of side information states, and let a genie exchange the current side information realization with the best side information in each discretization interval. Separation is optimal for the genie aided system, since it has a discrete number of side information states. In the limit of finer discretizations, the genie aided system can be shown to be achievable, similarly to\cite{ng2007minimum}, with a separates source and channel coding scheme.

First, we divide the side information state $\gamma $ into some partition
$\mathbf{s}$ given by $[s_0,s_1), [s_1,s_2),...$,
such that $s_0=0<s_1<...<s_i<\cdots$ and $\gamma  \in [s_{i-1},s_i)$ if
$s_{i-1}\leq \gamma <s_i$ for some $i=1,2,...$. The length of the
partition $[s_{i-1},s_i)$ is defined by $\Delta s_i$, i.e., $\Delta
s_i\triangleq s_i-s_{i-1}$. Let us define $\bar{\gamma} >0$ as the
super-level set $\bar{\gamma} $ satisfying (\ref{eq:OptGamma}). The
partition is chosen such that for some index $j$, we have
$s_{j}=\bar{\gamma} $. A fading realization belongs to the
interval $[s_{i-1}, s_i)$ with probability
$p_i=\int_{s_{i-1}}^{s_{i}}p_{\Gamma}(\gamma )d\gamma $.

We assume that when $\gamma $ belongs to the interval
$[s_{i-1},s_{i})$, a genie substitutes the current side information
sequence $Y=\gamma_c X+N$ with a sequence with gain $s_i$, i.e., $\tilde{Y}\triangleq \sqrt{s_i}X+N$.
Note that this receiver has a better performance as noise can be
added to $\tilde{Y}$ to recover a sequence equivalent to the original side information
sequence if required. Hence, the expected distortion for a given
partition $\mathbf{s}$, denoted by $ED^*_{gen}(\mathbf{s},\mathrm{C})$, is a
lower bound on the expected distortion of the continuous fading
setup. The genie aided system now consists of a countable number of
receivers. Due to the optimality of separation under countable number of side information states, $ED^*_{gen}(\mathbf{s},\mathrm{C})$ is given by the concatenation of a Heegard-Berger source encoder with side information states $s_1,s_2,...$, and a capacity achieving channel code, i.e., \mbox{
$ED^*_{gen}(\mathbf{s},\mathrm{C})=ED^*_{\mathbf{s}}(C)$, where }
\vspace{-2mm}
\begin{IEEEeqnarray}{rCl}
ED^*_{\mathbf{s}}(R)\triangleq \min_{\mathbf{D}:R_{\mathrm{HB}}(\mathbf{D})\leq R} \mathbf{p}^T\mathbf{D},
\end{IEEEeqnarray}
where $\mathbf{p}=[p_1,p_2,...]$ and $\mathbf{D}=[D_1,D_2,...]$ depend on the partition $\mathbf{s}$. This optimization problem is  studied in detail in \cite{ng2007minimum}.

%
Next, we consider an upper bound on $ED_{sta}^*$. With the channel state $h_c$ known at the encoder, we concatenate of a single layer source encoder for side information state $\bar{\gamma} $,  with a channel code at a rate arbitrarily close to the capacity $\mathrm{C}$. This scheme achieves an expected distortion of $ED_{Q}^*(\mathrm{C})$. Then,
\vspace{-5mm}
\begin{IEEEeqnarray}{rCl}\label{eq:InequalityCont}
ED^*_{gen}(\mathbf{s},\mathrm{C})\leq ED^*_{sta}\leq ED_{Q}^*(\mathrm{C}).
\end{IEEEeqnarray}

As the partition gets finer in the sense that $\max_i \Delta s_i\rightarrow 0$, it is shown in \cite{ng2007minimum} (see\cite[Proposition
4]{ng2007minimum} and \cite[Proposition 5]{ng2007minimum}) that 
$\lim _{\max_i \Delta s_i\rightarrow
 0}ED^*_{\mathbf{s}}(R)=ED_{Q}^*(R)$. Therefore, $\lim _{\max_i \Delta s_i\rightarrow
 0}ED^*_{gen}(\mathbf{s},\mathrm{C})=ED_{Q}^*(\mathrm{C})$, and from inequality
 (\ref{eq:InequalityCont}) in the limit of finer
 partitions, $ ED^*_{sta}=ED_{Q}^*(\mathrm{C})$. This completes the proof.

\section{Proof of Lemma \ref{lem:ConvergentBounds}}\label{app:ConvergentBounds}

In order to show the convergence of $ED_{pi}^*$ to $ED_{\text{inf}}$, first, we construct an upper bound on $ED_{pi}^*$ and we show that this bound converges to $ED_{\text{inf}}$ for large enough $L$.

The lower bound $ED_{pi}^*$ is achieved by the concatenation of a capacity achieving channel code with a single-layer source code targeting the side information state $\bar{\gamma} $, the solution to (\ref{eq:OptGamma}), for each realization of ${H}$. Instead, we consider that, for a given $L$ the source coding is done targeting the state
$\bar{\gamma}_{L}\triangleq \mu-\delta,$
where $\mu\triangleq \mathrm{E}[\Gamma_{L}]$ is the mean of $\Gamma_{L}$ and $\delta\triangleq\sqrt{\sigma^2_{L}}$. The expected distortion achieved by this scheme is an \mbox{ upper bound on $ED_{pi}^*$ and is found, similarly to $ED^*_{pi}$,  be given by}
\begin{IEEEeqnarray}{rCl}
ED_{lay}&\triangleq& \mathrm{E}_{{H}}\left[ED_{Q}\left(\frac{1}{2}\log(1+{H})\right)\right]\nonumber\\
&=&\int_{0}^{\bar{\gamma}_L}\frac{p_L(\gamma )}{1+\gamma }d\gamma +\int_{h}\int_{\bar{\gamma}_L}^{\infty}\frac{p_L(\gamma )p_{{H}}(h)}{(1+h)(1+\bar{\gamma}_L)+\gamma -\bar{\gamma}_L}d\gamma dh\nonumber,
\end{IEEEeqnarray}
where $ED_Q(R)$ is given as in (\ref{eq:SingleLayerCode}) for $\bar{\gamma}$ substituted by $\bar{\gamma}_L$ and $p_{L}(\gamma )$ is the pdf of $\Gamma_{L}$.

Then, we have the following bound
\begin{IEEEeqnarray}{rCl}
&ED_{pi^*}&-ED_{\text{inf}}\leq ED_{lay}-ED_{\text{inf}}\nonumber\\
&=&\int_{0}^{\bar{\gamma}_L}\frac{p_L(\gamma )}{1+\gamma }d\gamma +\int_{h}\int_{\bar{\gamma}_L}^{\infty}\frac{p_L(\gamma )p_{{H}}(h)}{(1+h)(1+\bar{\gamma}_L)+\gamma -\bar{\gamma}_L}d\gamma dh\nonumber
-\int_{h}\int_{\gamma }\frac{p_{{H}}(h)p_{L}(\gamma )}{(1+h)(1+\gamma )}d\gamma dh\nonumber\\
&\stackrel{(a)}{\leq}& \int_{0}^{\bar{\gamma}_L}p_L(\gamma )dh+\int_{h}\int_{\bar{\gamma}_L}^{\infty}\frac{p_L(\gamma )p_{{H}}(h)}{(1+h)(1+\bar{\gamma}_L)+\gamma -\bar{\gamma}_L}d\gamma dh\nonumber
-\int_{h}\int_{\mu-\delta}^{\mu+\delta}\frac{p_{{H}}(h)p_{L}(\gamma )}{(1+h)(1+\gamma )}d\gamma dh\nonumber\\
&=& \text{Pr}[\Gamma_{L}<\bar{\gamma}_L]+\int_{h}\int_{\bar{\gamma}_L}^{\mu+\delta}\frac{p_L(\gamma )p_{{H}}(h)}{(1+h)(1+\bar{\gamma}_L)+\gamma -\bar{\gamma}_L}d\gamma dh
\nonumber\\
&&+\int_{h}\int_{\mu+\delta}^{\infty}\frac{p_L(\gamma )p_{{H}}(h)}{(1+h)(1+\bar{\gamma}_L)+\gamma -\bar{\gamma}_L}d\gamma dh
\nonumber-\int_{h}\int_{\mu-\delta}^{\mu+\delta}\frac{p_{{H}}(h)p_{L}(\gamma )}{(1+h)(1+\gamma )}d\gamma dh\nonumber\\
&\stackrel{(b)}{\leq}& \text{Pr}[\Gamma_{L}<\bar{\gamma}_L]+\int_{h}\int_{\mu-\delta}^{\mu+\delta}\frac{p_L(\gamma )p_{{H}}(h)}{(1+h)(1+\bar{\gamma}_L)+\gamma -\bar{\gamma}_L}d\gamma dh
\nonumber\\
&&+\text{Pr}[\Gamma_{L}\geq\mu+\delta]-\int_{h}\int_{\mu-\delta}^{\mu+\delta}\frac{p_{{H}}(h)p_{L}(\gamma )}{(1+h)(1+\gamma )}d\gamma dh\nonumber\\
&\stackrel{(c)}{=}& \text{Pr}[|\Gamma_{L}-\mu|\leq \delta]
\nonumber+\int_{h}\int_{\mu-\delta}^{\mu+\delta}\frac{h(\gamma -\bar{\gamma}_L)p_L(\gamma )p_{{H}}(h)}{((1+h)(1+\bar{\gamma}_L)+\gamma -\bar{\gamma}_L)(1+h)(1+\gamma )}d\gamma dh\nonumber\\
&\stackrel{(d)}{\leq}&\text{Pr}[|\Gamma_{L}-\mu|\leq \delta]+ \mathrm{E}_{{H}}[{H}]\cdot 2\delta\nonumber\\
&\stackrel{(e)}{\leq}&\frac{\sigma^2_{L}}{\delta}+\mathrm{E}_{{H}}[{H}]\cdot 2\delta\nonumber
\end{IEEEeqnarray}
where $(a)$ follows since $\frac{1}{(1+\gamma )}\leq 1$ for  the first integral, and because we are reducing the integration region in the third one, $(b)$ follows due to  \begin{IEEEeqnarray}{rCl}
&\int_{h}&\int_{\mu+\delta}^{\infty}\frac{p_L(\gamma )p_{{H}}(h)}{(1+h)(1+\bar{\gamma}_L)+\gamma -\bar{\gamma}_L}d\gamma  dh\nonumber\\
&&\leq\int_{h}\int_{\mu+\delta}^{\infty}p_L(\gamma )p_{{H}}(h)d\gamma dh\nonumber\\
&&=\text{Pr}[\Gamma_{L}\geq\mu+\delta]\nonumber.
\end{IEEEeqnarray}
Then $(c)$ follows since $\bar{\gamma}_L=\mu-\delta$, and subtracting the two integrals, $(d)$ follows from the following bound,
\begin{IEEEeqnarray}{rCl}
&\int_{h}&\int_{\mu-\delta}^{\mu+\delta}\frac{h(\gamma -\bar{\gamma}_L)p_L(\gamma )p_{{H}}(h)}{((1+h)(1+\bar{\gamma}_L)+\gamma -\bar{\gamma}_L)(1+h)(1+\gamma )}d\gamma dh\nonumber\\
&\leq&\int_{h}\int_{\mu-\delta}^{\mu+\delta}h(\gamma -\bar{\gamma}_L)p_L(\gamma )p_{{H}}(h)d\gamma dh\nonumber\\
&\stackrel{(f)}{\leq}&\mathrm{E}[{H}]\cdot (\mu+\delta-\bar{\gamma}_L)\int_{\mu-\delta}^{\mu+\delta}p_L(\gamma )d\gamma \nonumber\\
&\stackrel{(g)}{\leq}& \mathrm{E}[{H}]\cdot 2\delta\nonumber
\end{IEEEeqnarray}
where $(f)$ follows since $\gamma \leq \mu+\delta$ in the integration region; $(g)$ follows since $\bar{\gamma}_L=\mu-\delta$ and $\int_{\mu-\delta}^{\mu+\delta}p_L(\gamma )d\gamma \leq 1$. Finally, $(e)$ follows from Chebyshev's inequality.

By  the choice of $\delta=\sqrt{\sigma_{L}^2}$, we have
\begin{IEEEeqnarray}{rCl}
ED_{pi}^*-ED_{\text{inf}}&\leq&\frac{\sigma^2_{L}}{\delta}+\mathrm{E}[{H}]\cdot 2\delta\nonumber
=\sqrt{\sigma_{L}^2}+\mathrm{E}[{H}]\cdot 2\sqrt{\sigma_{L}^2},\nonumber
\end{IEEEeqnarray}
and the difference converges to $0$ from the assumption $\sigma^2_{L}\rightarrow 0$ for $L\rightarrow \infty$. This completes the proof.

\section{Converse}\label{app:PartiallyInformedGeneral}

\subsection{Partially Informed Encoder  Upper Bound}\label{app:PartiallyInformedGeneral}
 In Section \ref{ssec:PartiallyInformedBound} we have seen that for continuous quasiconcave pdfs, $ED_{pi}^*$ is obtained by averaging the expected distortion achievable by the concatenation of a single layer source code designed for the side information state $\bar{\gamma}(h)$ and an optimal channel code for the current channel state $h$. For each $h$, the optimal $\bar{\gamma}(h)$ is determined by solving (\ref{eq:OptGamma}) with $R=\mathcal{C}(h)=\frac{1}{2}\log(1+h)$. Note that $\bar{\gamma}(h)$ is a random variable dependant on the realization of the channel fading ${H}$.

An upper bound on the distortion exponent can be found by lower bounding $ED_{pi}^*$. First, we note that $ED^*_Q(R)$ in (\ref{eq:SingleLayerCode}) is a convex function of $R$. This follows from the time-sharing arguments and convexity of the Heegard-Berger rate-distortion function \cite{Berger1985SideInfMayBeAbsent}. Then, by Jensen's inequality, we have
\begin{IEEEeqnarray}{rCl}
ED^*_{pi}=\mathrm{E}_{{H}}[ED_Q^*(\mathcal{C}({H}))]\geq ED_Q^*(\mathrm{E}_{{H}}[\mathcal{C}({H})]),
\end{IEEEeqnarray}
where
\begin{IEEEeqnarray}{rCl}\label{eq:FirstBound}
ED_Q^*(\mathrm{E}_{{H}}[\mathcal{C}({H})])=\int_{0}^{\tilde{\gamma} }\frac{p_{\Gamma}(\gamma)}{1+\gamma}d\gamma+\int_{\tilde{\gamma} }^{\infty}\!\frac{p_{\Gamma}(\gamma)}{(\tilde{\gamma} +1)2^{2\mathrm{E}_{{H}}[\mathcal{C}({H})]}+\gamma-\tilde{\gamma} }d\gamma,
\end{IEEEeqnarray}
and $\tilde{\gamma}$ is the solution to (\ref{eq:OptGamma}) with $R=\mathrm{E}_{{H}}[\mathcal{C}({H})]$, that is, the ergodic capacity of the channel. Note that $\tilde{\gamma}$ depends only on the ergodic capacity of the channel and not on the current channel state realization, and therefore, is not a random variable, as opposed to $\bar{\gamma}(h)$.

Now, since $\mathcal{C}(h)$  is a concave function of $h$, applying Jensen's inequality again, we have
\begin{IEEEeqnarray}{rCl}\label{eq:SecondBound}
\mathrm{E}_{{H}}[\mathcal{C}({H})]=\mathrm{E}_{{H}}\left[\frac{1}{2}\log(1+{H})\right]\leq\frac{1}{2}\log(1+\mathrm{E}[{H}])=\frac{1}{2}\log(1+\rho),  \end{IEEEeqnarray}
 that is, the ergodic capacity of the channel is lower than the capacity of a  static channel with the same average SNR.

We define, for $\hat{\gamma}\geq0$,
\begin{IEEEeqnarray}{rCl}\label{eq:Epigammatilde}
ED_{pe}(\hat{\gamma})&\triangleq&\int_{0}^{\hat{\gamma}} \frac{p_{\Gamma}(\gamma)}{1+\gamma}d\gamma+\int_{\hat{\gamma}}^{\infty}\frac{p_{\Gamma}(\gamma)}{(\hat{\gamma} +1) (1+\rho)+\gamma-\hat{\gamma} }d\gamma.
\end{IEEEeqnarray}
Then, we have
\begin{IEEEeqnarray}{rCl}\label{eq:EDUP}
ED^*_{pi}&\stackrel{(a)}{\geq} &\int_{0}^{\tilde{\gamma} }\frac{p_{\Gamma}(\gamma)}{1+\gamma}d\gamma+\int_{\tilde{\gamma}  }^{\infty}\!\frac{p_{\Gamma}(\gamma)}{(\tilde{\gamma}  +1) (1+\rho)+\gamma-\tilde{\gamma}  }d\gamma\nonumber\\
&\stackrel{(b)}{\geq}&\min_{\hat{\gamma}\geq0.}ED_{pe}(\hat{\gamma})\triangleq ED^*_{pe},
\end{IEEEeqnarray}
where $(a)$ follows from  inequality (\ref{eq:SecondBound}), and $(b)$ follows from the definition in (\ref{eq:Epigammatilde}).

Now, we obtain the exponential behavior of $ED^*_{pe}$.  Consider a sequence of normalized gamma distributed random variables ${H}_0\sim\Upsilon(L,\theta)$ under the change of variables
$A=-\frac{\log{H}_{0}}{\log \rho}$. The pdf for $A$ is found as
\begin{IEEEeqnarray}{rCl}
p_{A}(\alpha)&=&\left|\frac{\partial {H}_{0}}{\partial \alpha}\right|p_{{H}_{0}}(h_{0})=\rho^{-\alpha} p_{{H}_{0}}(\rho^{-\alpha})\log \rho.
\end{IEEEeqnarray}

Then, $p_{A}(\alpha)$ is given by
\begin{IEEEeqnarray}{rCl} p_{A}(\alpha)&=&\rho^{-\alpha}\frac{1}{\theta^{L}}\frac{1}{\Psi(L)}\rho^{-\alpha(L-1)}e^{-\frac{\rho^{-\alpha}}{\theta}}\log
\rho\nonumber
=\frac{1}{\theta^{L}}\frac{1}{\Psi(L)}\rho^{-L\alpha}e^{-\frac{\rho^{-\alpha}}{\theta}}\log\rho,
\end{IEEEeqnarray}
and the exponential behavior is found as
\begin{IEEEeqnarray}{rCl}\label{eq:RateFunctionSa}
S_{A}(\alpha)=-\lim_{\rho\rightarrow\infty}\frac{\log p_{A}(a)}{\log \rho}=\begin{cases}
L\alpha&\text{if }\alpha \geq 0,\\
+\infty& \text{if }\alpha <0.
\end{cases}
\end{IEEEeqnarray}

 For the model considered in Section \ref{sec:FiniteSNRResults}, the SNR exponent for the Nakagami fading channel, ${H}_{0}\sim\Upsilon(L_c,L_c^{-1})$, is given by $S_{A}(\alpha)=L_c\alpha$ for $\alpha\geq0$, and for the Nakagami fading side information, $\Gamma_{0}\sim\Upsilon(L_s,L_s^{-1})$, we have $S_{B}(\beta)=L_s\beta$ for $\beta\geq0$.

Define $\kappa\triangleq \frac{\log\hat{\gamma}}{\log\rho}$, such that $\hat{\gamma}=\rho^{\kappa}$.
Applying the change of variables to (\ref{eq:Epigammatilde}),  in the high SNR regime, we have
\begin{IEEEeqnarray}{rCl}\label{eq:EDPI}
ED_{pe}(\rho^{\kappa})
&=&\int_{\mathcal{A}_{pe}^{c}}\frac{p_B(\beta)}{(\rho^{\kappa}+1)(1+\rho)+\rho^{1-\beta}-\rho^{\kappa}} d \beta+\int_{\mathcal{A}_{pe}}\frac{p_B(\beta)}{1+\rho^{1-\beta}} d\beta\\
&\doteq&\!\int_{\mathcal{A}_{pe}^{c}}\rho^{-(\kappa^++1)}p_{B}(\beta) d \beta+\int_{\mathcal{A}_{pe}}\rho^{-(1-\beta)^+}p_{B}(\beta) d\beta\nonumber\nonumber
\end{IEEEeqnarray}
where we have defined
\begin{IEEEeqnarray}{rCl}
\mathcal{A}_{pe}&\triangleq&\{\beta:\hat{\gamma}\geq\rho^{1-\beta}\}\nonumber= \{\beta:\kappa \geq1-\beta\},
\end{IEEEeqnarray}
and we have used the fact that, in the high SNR asymptotic,
and for $\beta\in\mathcal{A}_{pe}^{c}$, we have
\begin{IEEEeqnarray}{rCl}
[(\rho^{\kappa}+1)(1+\rho)+\rho^{1-\beta}-\!\rho^{\kappa}]^{-1}\nonumber
&\stackrel{(a)}{\doteq}&[\rho^{\kappa^++1}+\rho^{1-\beta}-\rho^{\kappa}]^{-1}\\
&\stackrel{(b)}\doteq&\rho^{-\max\{\kappa^++1,1-\beta\}}\nonumber\\
&\stackrel{(c)}{=}&\rho^{-(\kappa^++1)},\nonumber
\end{IEEEeqnarray}
which $(a)$ and $(b)$ follows since $\rho^{x}+\rho^{y}\doteq\rho^{\max\{x,y\}}$ for $x, y\geq0$, and $(c)$ follows since  we have $1-\beta> \kappa$ for $\beta\in\mathcal{A}_{pe}^{c}$. 

In order to find the exponential behavior of the  $ED_{pe}(\rho^{\kappa})$, we study the exponent of each integral term in (\ref{eq:EDPI}). For the first term, we have
\begin{IEEEeqnarray}{rCl}
\Delta_{p1}(\kappa)&\triangleq&-\lim_{\rho\rightarrow\infty}\frac{1}{\log\rho}\log\int_{\mathcal{A}_{pe}}\rho^{-(1-\beta)^+}p_{B}(\beta) d\beta\nonumber\\
&\doteq&-\lim_{\epsilon\rightarrow 0}\epsilon\log\int_{\mathcal{A}_{pe}}\exp
\left(\frac{1}{\epsilon}(-[(1-\beta)^++S_{B}(\beta)])\right) d\beta\nonumber\\
&=&\inf_{\mathcal{A}^c_{pe}}\kappa^++1+S_{B}(\beta),\label{eq:DistExpoPi1}
\end{IEEEeqnarray}
where the last equality follows from Varadhan's Lemma \cite{Dembo:book}, similar to the proof of Theorem 4 in \cite{zheng2003diversity}. Similarly, for the second integral term in (\ref{eq:EDPI}), we have,
\begin{IEEEeqnarray}{rCl}\label{eq:DistExpoPi2}
\Delta_{p2}(\kappa)&\triangleq&\inf_{\mathcal{A}^c_{pe}}\kappa^++1+S_{B}(\beta).
\end{IEEEeqnarray}
We can lower bound (\ref{eq:EDUP}) as follows
\begin{IEEEeqnarray}{rCl}
ED_{pi}^*\geq\min_{\kappa\in\mathds{R}}\{ED_{pe}(\rho^{\kappa})\}\stackrel{.}{\geq}\min_{\kappa\in\mathds{R}}\{\rho^{-\Delta_{p1}(\kappa)}+\rho^{-\Delta_{p2}(\kappa)}\}\doteq\rho^{-\max_{\kappa\in\mathds{R}}\min\{\Delta_{p1}(\kappa),\Delta_{p2}(\kappa)\}}.
\end{IEEEeqnarray}
Then, the distortion exponent is upper bounded by
\begin{IEEEeqnarray}{rCl}\label{eq:DisteExpMinPartInfKappas}
-\lim_{\rho\rightarrow\infty}\frac{\log ED_{pi}^*}{\log\rho}\leq\max_{\kappa\in\mathds{R}}\min\{\Delta_{p1}(\kappa),\Delta_{p2}(\kappa)\}.
\end{IEEEeqnarray}

We solve the optimization problem in (\ref{eq:DisteExpMinPartInfKappas}) with $S_{B}(\beta)=L_s\beta$, and denote the optimal value by $\Delta_{pe}(L_s,L_c)$. We note that we can restrict the domain of $\beta$ in (\ref{eq:DistExpoPi1}) and (\ref{eq:DistExpoPi2}) to $\beta \geq 0$ without loss of optimality since $S_{B}(\beta)=+\infty$ for $\beta<0$.

First, we consider the case $\kappa<0$. In that case, $\Delta_{p1}(\kappa)$ is minimized by $\beta^*=1-\kappa$ and we have $\Delta_{p1}(\kappa)=L_s(1-\kappa)$. On the other hand, we have
\begin{IEEEeqnarray}{rCl}
\Delta_{p2}(\kappa)&=&\inf_{ \beta\geq 0}1+L_s\beta\nonumber\\
&&\text{s.t. }\beta< 1-\kappa,
\end{IEEEeqnarray}
which is minimized by $\beta^*=0$, and $\Delta_{p2}(\kappa)=1$. Then, from (\ref{eq:DisteExpMinPartInfKappas}), we have $\Delta_{pe}(L_s,L_c)=\max_{\kappa<0}\min\{L_s(1-\kappa),1\}$, which is maximized by $\kappa=-\infty$, and we have $\Delta_{pe}(L_s,L_c)=1$.

Next, we consider the case $\kappa\geq0$. Substituting $S_{B}(\beta)=L_s\beta$ in $\Delta_{p1}(\kappa)$ in (\ref{eq:DistExpoPi1}), we note that we can constrain our search to $0\leq \beta\leq 1$, since any $\beta>1$ can only increase the objective function. We have,
\begin{IEEEeqnarray}{rCl}
\Delta_{p1}(\kappa)&=&\inf_{0\leq\beta\leq 1}1+(L_s-1)\beta\nonumber\\
&&\text{s.t. } \beta\geq1-\kappa.
\end{IEEEeqnarray}
 Since for $L_s> 1$, $1+(L_s-1)\beta$ is increasing in $\beta$, the minimum is achieved by $\beta^*=(1-\kappa)^+$ and $\Delta_{p1}(\kappa)=1+(L_s-1)(1-\kappa)^+$. On the contrary, for $L_s\leq 1$, the objective function is decreasing in $\beta$, and is minimized at $\beta^*=1$, which yields $\Delta_{p1}(\kappa)=L_s$.

Similarly, for $\Delta_{p2}(\kappa)$ in (\ref{eq:DistExpoPi2}), we have
\begin{IEEEeqnarray}{rCl}
\Delta_{p2}(\kappa)&=&\inf_{ \beta\geq0}\kappa+1+L_s\beta\nonumber\\
&&\text{s.t. }\beta< 1-\kappa.
\end{IEEEeqnarray}
This problem is minimized by $\beta^*=0$, for which $\Delta_{p2}(\kappa)=1+\kappa$, for $0\leq \kappa<1$, and has no solution for $\kappa\geq1$, since there are no feasible $\beta$ in the optimization set.

Then, substituting in  (\ref{eq:DisteExpMinPartInfKappas}), for $L_s\leq1$, we have $\Delta_{pe}(L_s,L_c)=\max_{\kappa\geq0}\min\{L_s,1+\kappa\}=L_s$, and $\Delta_{pe}(L_s,L_c)=1$. For $L_s>1$, since $\Delta_{p1}(\kappa)$ is decreasing in $\kappa$  while $\Delta_{p2}(\kappa)$ is increasing in $\kappa$, the maximum $\Delta_{pe}(L_s,L_c)$ in (\ref{eq:DisteExpMinPartInfKappas}) is achieved when the two exponents are equal, i.e., $1+\kappa=1+(L_s-1)(1-\kappa)$, from which we find
\begin{IEEEeqnarray}{rCl}\label{eq:distexpkappapos}
\Delta_{pe}(L_s,L_c)=2-\frac{1}{L_s}, \quad \text{for } \kappa^*=\frac{L_s-1}{L_s}\in(0,1).
\end{IEEEeqnarray}

Now, we find the maximizing $\kappa$ for each $L_s$ regime to obtain $\Delta^*_{pe}(L_s,L_c)$. For $L_s\leq 1$, the distortion exponent is maximized by $\kappa=-\infty$ and $\Delta_{pe}(L_s,L_c)=1$, since $\Delta_{pe}(L_s,L_c)=L_s$ for any $\kappa\geq0$. On the contrary, for $L_s\geq1$, the distortion exponent is maximized as (\ref{eq:distexpkappapos}), while $\Delta_{pe}(L_s,L_c)=1$ if we consider $\kappa<0$.

Note that when $L_s\leq 1$, the side information gain distribution is monotonically decreasing. Then $\bar{\gamma}(h) =0$ for any $h$ from Proposition \ref{prop:MonotonicallyDecresaing}, and therefore, from Theorem \ref{th:OptUnc}, uncoded transmission achieves the minimum expected distortion, i.e., $ED_{pi}^*=ED_u$. The distortion exponent for uncoded transmission $\Delta_u(L_s,L_c)$ is calculated in Appendix \ref{app:DistExpUnc} as $\Delta_{u}(L_s,L_c)=\min\{1,L_s+L_c\}$. Comparing  $\Delta_{u}(L_s,L_c)$ with $\Delta_{pe}(L_s,L_c)$, we observe that the proposed lower bound on $ED_{pi}^*$ is in general not tight due to inequality (\ref{eq:SecondBound}).

%

\section{Distortion Exponent Derivations}\label{app:DiestExpDerivations}

%

\subsection{Separate Source and Channel Coding (SSCC)}
Here we find the distortion exponent of SSCC. Let us define the events
\begin{IEEEeqnarray}{rCl}
\mathcal{O}_1&\triangleq& \{(h,\gamma ):R_c\geq I(X;Y)\},\nonumber\\
\mathcal{O}_2&\triangleq& \{(h,\gamma ):R_c<I(X;Y),R_c\leq I(S;W|T)\}\nonumber.
\end{IEEEeqnarray}
Event $\mathcal{O}_1$ corresponds to an outage due to bad quality of the channel, and $\mathcal{O}_2$ corresponds to a correct decoding of the channel codeword while an outage occurs due to the bad quality of the side information. It is readily seen that $\mathcal{O}_{s}=\mathcal{O}_{1}\bigcup\mathcal{O}_{2}$. Consider the change of variables ${H}_{0}=\rho^{-A}$, $\Gamma_{0}=\rho^{-B}$, $R_s=\frac{r_s}{2}\log\rho$ and $R_c=\frac{r_c}{2} \log\rho$, for $r_s\geq0$ and $r_c>0$. Note that we consider $r_s=0$ to allow SSCC to transmit without binning. We have
\begin{IEEEeqnarray}{rCl}
ED_{s}(R_c,R_s)\nonumber
&=&\!\int_{\mathcal{O}_{s}^c}\frac{p_{{H}}(h)p_{\Gamma}(\gamma )}{2^{2(R_c+R_s-\epsilon)}+\gamma }dhd\gamma +\int_{\mathcal{O}_{s}}\frac{p_{{H}}(h)p_{\Gamma}(\gamma )}{1+\gamma }dhd\gamma \nonumber\\
&=&\!\int_{\mathcal{A}_{s}^c(\rho)}\frac{p_{A}(\alpha)p_{B}(\beta)}{\rho^{r_c+r_s}+\rho^{1-\beta}}d\alpha d\beta+\int_{\mathcal{A}_{s}(\rho)}\frac{p_{A}(\alpha)p_{B}(\beta)}{1+\rho^{1-\beta}}d\alpha d\beta\nonumber,
\end{IEEEeqnarray}
where we have defined $\mathcal{A}_{s}(\rho)\triangleq \mathcal{A}_1(\rho)\bigcup\mathcal{A}_2(\rho)$,  and $\mathcal{A}_1(\rho)$ characterizes $\mathcal{O}_1$ in terms of $\alpha$ and $\beta$, and is given by
\begin{IEEEeqnarray}{rCl}
\mathcal{A}_1(\rho)&\triangleq &\left\{(h,\gamma):R_c\geq\frac{1}{2}\log(1+h)\right\}\nonumber=\{(\alpha,\beta):\rho^{r_c}\geq1+\rho^{1-\alpha}\}\nonumber,
\end{IEEEeqnarray}
and similarly for  $\mathcal{O}_2$ we have
\begin{IEEEeqnarray}{rCl}
\mathcal{A}_2(\rho)&\triangleq& \left\{(h,\gamma ):R_c<\frac{1}{2}\log(1+h),\, R_c\leq\frac{1}{2}\log\left(1+\frac{2^{2(R_s+R_c-\epsilon)}-1}{1+\gamma }\right)\right\}\nonumber\\
&=&\left\{(\alpha,\beta):\rho^{r_c}<1+\rho^{1-\alpha}, \, \rho^{r_c}\leq1+\frac{2^{-2\epsilon}\rho^{r_s+r_c}}{1+\rho^{(1-\beta)}}\right\}\nonumber.
\end{IEEEeqnarray}

Using similar bounding techniques to the ones used in Appendix \ref{app:PartiallyInformedGeneral}, it is not hard to show that in the high SNR regime, we have
\begin{IEEEeqnarray}{rCl}
ED_{s}(R_c,R_s)
&\doteq&\int_{\mathcal{A}_1^c\cap\mathcal{A}_2^c}\quad\frac{p_{A}(\alpha)p_B(\beta)}{\rho^{\max\{r_c+r_s,1-\beta\}}}d\alpha d\beta+\int_{\mathcal{A}_1\cup\mathcal{A}_2}
\frac{p_{A}(\alpha)p_B(\beta)}{\rho^{(1-\beta)^+}}d\alpha d\beta\nonumber,
\end{IEEEeqnarray}

where the equivalent outage sets in the high SNR are
\begin{IEEEeqnarray}{rCl}
\mathcal{A}_{1}&\triangleq& \{(\alpha,\beta): r_{c}\geq(1-\alpha)^+\},\nonumber\\
\mathcal{A}_{2}&\triangleq&\{(\alpha,\beta): r_{c}<(1-\alpha)^+, r_c\leq({r_s+r_c}-{(1-\beta)^+})^+\}.\nonumber
\end{IEEEeqnarray}

Let $\mathbf{r}\triangleq [r_c,r_s]$. Applying Varadhan's lemma, the distortion exponent of each integral term  are found as
\begin{IEEEeqnarray}{rCl}
\Delta_{s1}(\mathbf{r})&=&\inf_{\mathds{R}^2}\max\{r_c+r_s,1-\beta\}+S_{A}(\alpha)+S_{B}(\beta)\nonumber\\
&&\text{s.t. }r_{c}<(1-\alpha)^+, \quad r_c>({r_s+r_c}-{(1-\beta)^+})^+,\nonumber
\end{IEEEeqnarray}
and
\begin{IEEEeqnarray}{rCl}
\Delta_{s2}(\mathbf{r})&=&\inf_{\mathds{R}^2}(1-\beta)^++S_{A}(\alpha)+S_{B}(\beta)\\
&&\text{s.t. }r_{c}\geq(1-\alpha)^+,\nonumber\\
&&\text{or }r_{c}<(1-\alpha)^+, \quad r_c\leq({r_s+r_c}-{(1-\beta)^+})^+.\nonumber
\end{IEEEeqnarray}
We can limit the optimization to $0\leq\alpha,\beta\leq1$ without loss of optimality.
First, we find the distortion exponent for $L_s\geq1$. We start with $\Delta_{s1}(\mathbf{r})$. If $r_c+r_s\geq1-\beta$, we have
\begin{IEEEeqnarray}{rCl}
\Delta_{s1}(\mathbf{r})&=&\inf_{\alpha,\beta\geq0}r_s+r_c+L_c\alpha+L_s\beta\\
&&\text{s.t. }\alpha< 1-r_c, \quad 1-(r_s+r_c)\leq\beta < 1-r_s.\nonumber
\end{IEEEeqnarray}
The minimum is achieved by $\beta^*=(1-(r_s+r_c))^+$ and $\alpha^*=0$ and we have $\Delta_{s1}(\mathbf{r})=r_s+r_c+L_s(1-(r_s+r_c))^+$ for $r_c< 1$, $r_s< 1$. If $1-\beta> r_c+r_s$,
\begin{IEEEeqnarray}{rCl}
\Delta_{s1}(\mathbf{r})&=&\inf_{\alpha,\beta\geq0}1+L_c\alpha+(L_s-1)\beta\\
&&\text{s.t. }\alpha< 1-r_c, \quad \beta < 1-(r_s+r_c).\nonumber
\end{IEEEeqnarray}
The minimum is achieved by $\alpha^*=\beta^*=0$, and is found to be $\Delta_1(\mathbf{r})=1$ for $r_c<1$ and $r_c+r_s<1$. Then, putting all together, the infimum is given by $\Delta_{s1}(\mathbf{r})=\max\{1,r_s+r_c\}$, for  $r_s<1$ and $r_c<1$.

For $\Delta_{s2}(\mathbf{r})$, we first consider the case with constraint $r_c\geq(1-\alpha)^+$. The minimum is easily seen to be given by $\alpha^*=(1-r_c)^+$ and $\beta^*=0$. Then $\Delta_{s2}(\mathbf{r})=1+L_c(1-r_c)^+$.
If $r_c\leq(1-\alpha)^+$, the second constraint is active. If $r_s+r_c<(1-\beta)^+$, $\Delta_{s2}(\mathbf{r})$ has no solution since this would require $r_c\leq0$. If $r_s+r_c\geq(1-\beta)^+$, the minimum is achieved for $\alpha^*=0$ and $\beta^*=(1-r_s)^+$, and is given by $\Delta_{s2}(\mathbf{r})=1+(L_s-1)(1-r_s)^+$ for $r_s>0$ and $r_c<1$.

The optimal distortion exponent of SSCC can be found by maximizing over the rates as
\begin{IEEEeqnarray}{rCl}\label{eq:DisteExpMinPartInf}
\Delta_{s}(L_s,L_c)=\max_{r_{c},r_{s}\geq0}\min\{\Delta_{s1}(\mathbf{r}),\Delta_{s2}(\mathbf{r})\}.\nonumber
\end{IEEEeqnarray}
The distortion exponent is maximized when $r_s+r_c>1$, $r_c<1$ and $r_s<1$. Then, we have $\Delta_{s1}(\mathbf{r})=r_s+r_c$, $\Delta_{s2}(\mathbf{r})=\min\{1+L_c(1-r_c)^+,1+(L_s-1)(1-r_s)^+\}$. The maximum is achieved by $r_c$ and $r_s$ for which the left and right terms in the minimization in $\Delta_{s2}(\mathbf{r})$ are equal, i.e., $1+L_c(1-r_c)=1+(L_s-1)(1-r_s)$, and $\Delta_{s1}(\mathbf{r})=\Delta_{s2}(\mathbf{r})$. Solving this, we have
\begin{IEEEeqnarray}{rCl}
r_s^*=\frac{(L_c+1)(L_s-1)}{L_s(L_c+1)-1},\qquad r_c^*=\frac{L_cL_s}{L_s(L_c+1)-1},\nonumber
\end{IEEEeqnarray}
which satisfy $r_s<1$, $r_c<1$ and $r_s+r_c>1$. Note that for $L_s=1$, we have $r_s=0$, i.e., no binning is optimal, as expected from Lemma \ref{lem:OptLay}.

Now we consider the case $L_s\leq 1$. In this regime, the gamma function is monotonically decreasing, and hence, $\bar{\gamma} =0$ and from Lemma \ref{lem:OptLay} we have $R^*_s=0$, i.e., no binning achieves the minimum distortion for SSCC.  The distortion exponent achievable without binning follows similarly by observing that by letting $R_s=0$, the outage event $\mathcal{A}_2$ is empty.

%

\subsection{Joint Decoding Scheme (JDS)}

Here, we consider the distortion exponent for JDS. Applying the change of variables, ${H}_{0}=\rho^{-A}$, $\Gamma_{0}=\rho^{-B}$ and $R_j=\frac{r_j}{2}\log\rho$ for $r_h>0$, form (\ref{eq:Dj}) we have
\begin{IEEEeqnarray}{rCl}
ED_{j}(R_j)\nonumber
&=&\int_{\mathcal{O}_j^c}\frac{p_{{H}}(h)p_{\Gamma}(\gamma )}{2^{2(R_j-\epsilon)}+\gamma }dhd\gamma \!+\!
\int_{\mathcal{O}_j}\frac{p_{{H}}(h)p_{\Gamma}(\gamma )}{1+\gamma }dhd\gamma \nonumber\\
&\doteq&\int_{\mathcal{A}_j^{c}}\frac{p_{A}(\alpha)p_{B}(\beta)}{\rho^{\max\{r_j,(1-\beta)^+\}}}d\alpha d\beta+\int_{\mathcal{A}_j}\frac{p_{A}(\alpha)p_{B}(\beta)}{\rho^{(1-\beta)^+}}d\alpha d\beta,\nonumber
\end{IEEEeqnarray}
where we define the outage event in the high SNR regime as
\begin{IEEEeqnarray}{rCl}
\mathcal{A}_j\triangleq\left\{(\alpha,\beta):(r_j-(1-\beta)^+)^+\geq(1-\alpha)^+\right\}\nonumber.
\end{IEEEeqnarray}

The distortion exponent for each term is found applying Varadhan's Lemma as
\begin{IEEEeqnarray}{rCl}
\Delta_{j1}(r_j)&=&\inf_{\mathcal{A}_j^c}\max\{r_j,(1-\beta)^+\}+S_{A}(\alpha)+S_{B}(\beta)\nonumber,
\end{IEEEeqnarray}
and
\begin{IEEEeqnarray}{rCl}
\Delta_{j2}(r_j)&=&\inf_{\mathcal{A}_j}(1-\beta)^++S_{A}(\alpha)+S_{B}(\beta)\nonumber.
\end{IEEEeqnarray}

First we note that in both $\Delta_{j1}(r_j)$ and $\Delta_{j2}(r_j)$ we can restrict to $0\leq \alpha,\beta\leq 1$ without loss of optimality since $S_{A}(\alpha)=L_c\alpha$ and $S_{B}(\beta)=L_s\beta$. Now we solve $\Delta_{j1}(r_j)$. If $r_j<1-\beta$, we have $\mathcal{A}_j=\{(\alpha,\beta):(1-\alpha)^+\geq 0,r_j<1-\beta\}$ and it is easily seen that $\alpha^*=0$. Then if $L_s\geq 1$, we have $\beta^*=0$ and $\Delta_{j1}(r_j)=1$ for $r_j\leq 1$. If $L_s<1$, then $\beta^*=(1-r_j)^+$ and $\Delta_{j1}(r_j)=1+(L_s-1)(1-r_j)^+$ for $r_j\leq 1$. If $r_j\geq 1-\beta$, we have
\vspace{-0.3cm}
\begin{IEEEeqnarray}{rCl}
\Delta_{j1}(r_j)&=&\inf_{0\leq \alpha,\beta\leq 1} r_j+L_c\alpha+L_s\beta\nonumber\\
&&\alpha+\beta< 2-r_j, \quad \beta\geq1-r_j.
\end{IEEEeqnarray}
The minimum is achieved by $\alpha^*=0$ and $\beta^*=(1-r_j)^+$  if $r_j\leq 2$ and is given by $\Delta_{j1}(r_j)=r_j+L_s(1-r_j)^+$ and has no feasible solutions if $r_j\geq 2$. Then, the exponent $\Delta_{j1}(r_j)$ is given by the minimum of these solutions, given by
\begin{IEEEeqnarray}{rCl}
\Delta_{j1}(r_j)=
\begin{cases}
1+(L_s-1)^+(1-r_j)&\text{if }0\leq r_j< 1,\\
r_j& \text{if }1\leq r_j< 2,
\end{cases}
\end{IEEEeqnarray}
where we have used that for $L_s\leq 1$ and $0\leq r_j \leq 1$, we have $r_j+L_s(1-r_j)^+=1+(1-L_s)^+(1-r_j)^+$, and for $L_s\geq 1$ and $0\leq r_j \leq 1$, we have $\min\{r_j+L_s(1-r_j)^+,1\}=1$.

Now, we solve $\Delta_{j2}(r_j)$. If $r_{j}<1-\beta$, the problem has no feasible solution due to the constraints. If $r_{j}\geq1-\beta$, we have
\vspace{-0.5cm}
\begin{IEEEeqnarray}{rCl}
\Delta_{j2}(r_j)&=&\inf_{0\leq \alpha,\beta\leq1} 1+(L_s-1)\beta+L_c\alpha\nonumber\\
&&\alpha+\beta\geq 2-r_j, \quad \beta\geq1-r_j.
\end{IEEEeqnarray}
The minimum is achieved by $\alpha^*=(2-r_j-\beta)^+$, which satisfies $\alpha^*\leq 1$ due to $\beta\geq1-r_j$. Then, if $\beta\geq2-r_j$ and $L_s\geq 1$, we have $\beta^*=(2-r_j)^+$ for $r_j\geq 1$ and the minimum is given by $\Delta_{j2}(r_j)=1+(L_s-1)(2-r_j)^+$.  If $\beta\geq2-r_j$ and $L_s< 1$ we have $\beta^*=1$ and $\Delta_{j2}(r_j)=L_s$ for $r_j\geq 1$. If $\beta< 2-r_j$ and $L_s\geq 1+L_c$, the minimum is achieved by $\beta^*=(1-r_j)^+$ if $r_j\leq 2$ and $\Delta_{j2}(r_j)=1+(L_s-1-L_c)(1-r_j)^++L_c(2-r_j)$. If $L_s<1+L_c $, the solution is found as $\Delta_{j2}(r_j)=L_s+L_c(1-r_j)$ if $r_j\leq 1$ for $\beta^*=1$ and by $\Delta_{j2}(r_j)=1+(L_s-1)(2-r_j)$ if $ r_j\geq1$ for $\beta=(2-r_j)^+-\delta$, for arbitrarily small $\delta>0$.

Finally, $\Delta_{j2}(r_j)$ is found by the minimum of these solutions in each regime. If $0\leq r_j\leq 1$, we have
\begin{IEEEeqnarray}{rCl}
\Delta_{j2}(r_j)=\begin{cases}
L_s+L_c(1-r_j)&\text{if } L_s< L_c+1,\\
1+L_c+(L_s-1)(1-r_j)&\text{if } L_s\geq L_c+1.
\end{cases}
\end{IEEEeqnarray}
If $1\leq r_j\leq 2$, we have
\begin{IEEEeqnarray}{rCl}
\Delta_{j2}(r_j)=\begin{cases}
L_s&\text{if } L_s<1,\\
1+\min\{L_c,L_s-1\}(2-r_j)^+&\text{if } L_s\geq 1,
\end{cases}
\end{IEEEeqnarray}
where for the case $L_s< 1$ we have that $L_s\leq L_s+L_c(1-r_j)$, and in the case $L_s\geq1$, we have that $1+L_c(2-r_j)\leq 1+(L_s-1)(2-r_j)$ for $L_s\geq 1+L_c$.
Finally, for $r_j\geq 2$ we have $\Delta_{j2}(r_j)=\min\{1,L_s\}$.

The distortion exponent can be maximized over $r_j$. If $L_s\leq 1$, the maximum is found by using a rate $0\leq r_j\leq 1$ and equating $\Delta_{j1}(r_j)=1+(L_s-1)(1-r_j)$ and $\Delta_{j2}(r_j)=L_s+L_c(1-r_j)$. The optimal rate is found as $r_j^*=\frac{L_c}{1+L_c-L_s}\leq 1$. If $1<L_s\leq L_c+1$, the maximum distortion exponent is found with a rate $1\leq r_j\leq 2$ such that $\Delta_{j1}(r_j)=r_j$ and $\Delta_{j2}(r_j)=1+(L_s-1)(2-r_j)$ are equal, given by $r_j^*=2-\frac{1}{L_s}$. Finally, if $L_s>L_c+1$, the distortion exponent is maximized when $1\leq r_j\leq 2$. By equaling $\Delta_{j1}(r_j)=r_j$ and  $\Delta_{j2}(r_j)=1+L_c(2-r_j)$, the distortion exponent is maximized by $r_j^*=1+\frac{L_c}{L_c+1}$.

\subsection{Superposed Hybrid Digital-Analog Transmission (S-HDA)}
The performance of the S-HDA scheme in Section \ref{ssec:S-HDA} can be  optimized over $P_d$, $P_a$ and $\eta^2$. From the distortion exponent perspective, we have observed that it suffices to allocate all the power to the digital component, which reduces S-HDA to HDA. Therefore, we let $P_d=1$, $P_a=0$. Then, applying the change of variables, we have from (\ref{eq:DistS-HDA2})-(\ref{eq:DistS-HDA}),
\begin{IEEEeqnarray}{rCl}
  ED_{shda}(1,\eta)&=&E_{\mathcal{O}_h}[D^{out}_{h}(\eta,1)]+E_{\mathcal{O}^c_h}[D_{h}(\eta,1)]\nonumber\\
  &=&\int_{\mathcal{O}_h}\frac{p_{{H}}(h)p_{\Gamma}(\gamma )}{1+\gamma }dhd\gamma \nonumber
  +\int_{\mathcal{O}_h^c}\frac{p_{{H}}(h)p_{\Gamma}(\gamma )}{1+\gamma +\eta^2(1+h)}dhd\gamma \nonumber\\
  &=&\int_{\mathcal{A}_h(\rho)} \frac{p_{A}(\alpha) p_{B}(\beta)}{1+\rho^{1-\beta}}d\alpha d\beta\label{eq:hybDistEq1}
  \!+\!\!\int_{\mathcal{A}_h^c(\rho)}\frac{p_{A}(\alpha)p_B(\beta)}{1+\rho^{1-\beta}+\eta^2(1+\rho^{1-\alpha})}d\alpha d\beta\nonumber,
    \end{IEEEeqnarray}
    where $\mathcal{O}_h$ in (\ref{eq:Outagehybrid}) is found, in terms of $\alpha$ and $\beta$ as
\begin{IEEEeqnarray}{rCl}
\mathcal{A}_h(\rho)&\triangleq&
\left\{(\alpha,\beta):\frac{\rho^{1-\alpha}}{1+\rho^{1-\alpha}}(1+\rho^{1-\beta})\leq \eta^2\right\}\nonumber.
\end{IEEEeqnarray}

In the high SNR regime, we let $\eta^2=\rho^{r_h}$, for $r_h\in\mathds{R}$ ,and the outage event $\mathcal{A}_h(\rho)$  is equivalent to
\begin{IEEEeqnarray}{rCl}\label{eq:SetHighHDA1}
\mathcal{A}_{h} \triangleq \left\{(\alpha,\beta):(1-\beta)^+-(\alpha-1)^+\leq r_h \right\}.
\end{IEEEeqnarray}

Then, we have
\begin{IEEEeqnarray}{rCl}\label{eq:uppBound}
&ED&_{shda}(1,\rho^{r_h}) \nonumber\\
\!&\doteq&\!\!\int_{\mathcal{A}_{h}} \!\rho^{-(1-\beta)^+} p_{A}(\alpha) p_{B}(\beta) d\alpha d\beta
+\int_{\mathcal{A}^c_{h}}\! \rho^{-\max\{(1-\beta)^+,(1-\alpha)^++r_h\}} p_{A}(\alpha) p_{B}(\beta) d\alpha d\beta.
  \end{IEEEeqnarray}

Using Varadhan's Lemma, the distortion exponent for the first integral in (\ref{eq:uppBound}) is found as
\begin{IEEEeqnarray}{rCl}
\Delta_{h1}(r_h)&\triangleq&\inf_{\mathcal{A}_{h}}(1-\beta)^++S_{A}(\alpha)+S_{B}(\beta)\nonumber,
\end{IEEEeqnarray}
and for the second integral as
\begin{IEEEeqnarray}{rCl}
\Delta_{h2}(r_h)&\triangleq&\inf_{\mathcal{A}_{h}^c}\max\{(1-\beta)^+,(1-\alpha)^++r_h\}+S_{A}(\alpha)+S_{B}(\beta)\nonumber.
\end{IEEEeqnarray}
The distortion exponent for HDA can be optimized over the parameter $r_h$ as
\begin{IEEEeqnarray}{rCl}\label{eq:S-HDAExpProblem}
\Delta_{hda}(L_s,L_c)=\max_{r_h\in \mathds{R}}\min\{\Delta_{h1}(r_h),\Delta_{h2}(r_h)\}.
\end{IEEEeqnarray}

First, we obtain the achievable distortion exponent when $r_h<0$. To solve $\Delta_{h1}(r_h)$, note that if $0\leq \alpha\leq 1$, there are no feasible solutions. Then, for $\alpha>1$, we have
\begin{IEEEeqnarray}{rCl}
\Delta_{h1}(r_h)&\triangleq&\inf_{\alpha>1,\beta\geq0}(1-\beta)^++L_c\alpha+L_s\beta\nonumber\\
&&\text{s.t. }\alpha\geq(1-\beta)^++1-r_h.
\end{IEEEeqnarray}
We can constrain the optimization to $0\leq \beta\leq 1$ without loss of optimality, and the minimum is achieved by $\alpha^*=2-\beta-r_h$. If $L_s\geq1+L_c$, the minimum is achieved by $\beta^*=0$, and is given by $\Delta_{h1}(r_h)=1+L_c(2-r_h)$. On the other hand, if $L_s<1+L_c$, $\beta^*=1$, and $\Delta_{h1}(r_h)=L_s+L_c(1-r_h)$. Putting all together, we have $\Delta_{h1}(r_h)=\min\{L_s,1+L_c\}+L_c(1-r_h)$.

Now, we solve $\Delta_{h2}(r_h)$. Without loss of optimality, we can assume $0\leq \alpha,\beta\leq 1$, as otherwise the feasible grows and $\alpha>1$ or $\beta>1$ can only increase the objective function. Then, the constraint is always satisfied, since $1-\beta\geq r_h$ for any $0\leq\beta\leq1$. We have
\begin{IEEEeqnarray}{rCl}
\Delta_{h2}(r_h)=\max_{0\leq\alpha,\beta\leq1}\{1-\beta,1-\alpha+r_h\}+L_s\beta+L_c\alpha.
\end{IEEEeqnarray}
If $1-\beta\geq1-\alpha+r_h$, the minimum is achieved by $\alpha^*=\beta^*=0$ when $L_s\geq1$ and $\Delta_{h2}(r_h)=1$. If $L_s<1$, $\beta^*=\alpha-r_h$ if $\alpha-r_h\leq 1$, and $\alpha^*=0$ when $L_s+L_c\geq 1$  and we have $\Delta_{h2}(r_h)=1-(L_s-1)r_h$. When $L_s+L_c< 1$, we have $\alpha^*=1+r_h$ and $\Delta_{h2}(r_h)=L_s+L_c(1+r_h)$, $-1\leq r_h<0$ and, when $\alpha>1+r_h$, we have $\beta^*=1$ and $\Delta_{h2}(r_h)=L_s+L_c(1+r_h)^+$. If $1-\beta< 1-\alpha+r_h$, we have $\beta^*=\alpha+\delta$, which has to satisfy $\beta^*\leq 1$, i.e., it is feasible whenever $\alpha\leq 1+r_h$. Then, $\alpha^*=0$ if $L_s+L_c\geq1$ and the minimum is given by $\Delta_{h2}(r_h)=1-r_h(L_s-1)$. If $L_s+L_c<1$, we have $\alpha^*=1+r_h$ and $\Delta_{h2}(r_h)=L_s+L_c(1+r_h)$, for $r_h\geq-1$. Putting all together, we have $\Delta_{h2}(r_h)=1$ when $L_s\geq1$ and $\Delta_{h2}(r_h)=\min\{1-(L_s-1)r_h,L_s+L_c(1+r_h)\}$ for $L_s<1$.

 If $L_s\leq 1$, we have $\Delta_{h1}(r_h)\geq\Delta_{h2}(r_h)$, and the distortion exponent is maximized by letting $r_h\rightarrow 0$ and we get $\Delta_{hda}(L_s,L_c)=\min\{L_s+L_c,1\}$.  If $L_s\geq1$, we have $\Delta_{hda}(L_s,L_c)=1$ for any $r_h<0$.

In the following, we derive the distortion exponent achievable by S-HDA when $r_h\geq0$. First, we solve $\Delta_{h1}(r_h)$. We can limit the optimization to $0\leq \beta\leq 1$ without loss of optimality. Then, for $0\leq \alpha\leq 1$ the minimum is achieved by $\alpha^*=0$, and if $L_s\geq1$, the minimum is achieved by $\beta^*=(1-r_h)^+$ and $\Delta_{h1}(r_h)=1+(L_s-1)(1-r_h)^+$, and if $L_s<1$, $\beta^*=1$ and $\Delta_{h1}(r_h)=L_s$. If $\alpha>1$, the constraint becomes $\alpha\geq2-\beta-r_h$, and the minimizing $\alpha$ is given by $\alpha^*=2-\beta-r_h$, which is feasible provided that $\beta<1-r_h$. Then, we have
\begin{IEEEeqnarray}{rCl}\label{eq:HDADistExp2}
\Delta_{h1}(r_h)&=&\inf_{0\leq \beta\leq 1 }1+(L_s-1-L_c)\beta+L_c(2-r_h)\nonumber\\
&&\text{s.t. }\beta<1-r_h.
\end{IEEEeqnarray}
If $L_s\geq1+L_c$, we have $\beta^*=0$ and $\Delta_{h1}(r_h)=1+L_c(2-r_h)$ for $r_h\leq 1$, and if $L_s<1+L_c$, we have $\beta^*=1-r_h$ and $\Delta_{h1}(r_h)=1+L_c+(L_s-1)(1-r_h)$. Putting all together, $\Delta_{h1}(r_h)$ is found as
\begin{IEEEeqnarray}{rCl}
\Delta_{h1}(r_h)=\begin{cases}
L_s&\text{if }L_s<1,\\
1+(L_s-1)(1-r_h)^+ &\text{if }L_s\geq1.
\end{cases}
\end{IEEEeqnarray}

Next, we solve $\Delta_{h2}(r_h)$. First, we note that we can constrain to $0\leq\beta\leq 1$, since the optimization set is empty if $\beta>1$. Similarly, we assume $0\leq\alpha\leq 1$, since any $\alpha>1$ achieves a larger exponent. Then,
\begin{IEEEeqnarray}{rCl}\label{eq:HDADistExp2}
\Delta_{h2}(r_h)&=&\inf_{0\leq \alpha,\beta\leq 1 }\max\{1-\beta,1-\alpha+r_h\}+L_s\beta+L_c\alpha\nonumber\\
&&\text{s.t. }\beta<1-r_h.
\end{IEEEeqnarray}

If $1-\beta>1-\alpha+r_h$, we have $\alpha^*=\beta+r_h$, which satisfies $\alpha^*\leq1$ since $\beta<1-r_h$. Then, $\beta^*=0$ if $L_s+L_c\geq1$ and $\Delta_{h2}(r_h)=1+L_cr_h$, and if $L_s+L_c<1$, $\beta^*=1-r_h-\epsilon$ for an arbitrarily $\epsilon>0$ and the infimum is found as $\Delta_{h2}(r_h)=1+L_c+(L_s-1)(1-r_h)$ for $r_h<1$. If $1-\beta\leq 1-\alpha+r_h$, the infimum is given by $\beta^*=(\alpha-r_h)^+$. If $\alpha\geq r$ and $L_s+L_c\geq1$, the minimum is found as $\alpha^*=r_h$ and $\Delta_{h2}(r_h)=1+r_hL_c$, while $\alpha^*=1$ if $L_s+L_c<1$, and $\Delta_{h2}(r_h)=1+L_c+(L_s-1)(1-r_h)$. If $\alpha<r_h$, we have $\alpha^*=0$ if $L_c\geq1$ and $\Delta_{h2}(r_h)=1+r_h$ and if $L_c<1$, we have $\alpha^*=r_h+\epsilon$ for an arbitrarily small $\epsilon>0$ and $\Delta_{h2}(r_h)=1+r_hL_c$. Putting all together, we have $\Delta_{h2}(r_h)=1+\min\{1,L_c\}r_h$ for $r_h\leq 1$.

We optimize over $r_h$ to solve (\ref{eq:S-HDAExpProblem}).
For $L_s\leq 1$, we have $\Delta_{h1}(r_h)<\Delta_{h2}(r_h)$ for any $r_h\geq 0$ and $\Delta_{hda}(L_s,L_c)=L$. Then, the achievable distortion exponent is maximized, by using $r_h<0$ and $r_h\rightarrow0$, for which we obtain $\Delta_{hda}(L_s,L_c)=\min\{L_s+L_c,1\}$. On the contrary, when $L_s\geq1$, the distortion exponent is maximized for an $r_h>0$ such that $\Delta_{h1}(r_h)=\Delta_{h2}(r_h)$, i.e.,
\begin{IEEEeqnarray}{rCl}
r_h^*=\frac{(L_s-1)}{L_s-1+\min\{1,L_c\}}.
\end{IEEEeqnarray}
Putting all together we obtain the achievable distortion exponent in  (\ref{eq:S-HDAexponent}).

\bibliographystyle{ieeetran}
\bibliography{ref}

\begin{thebibliography}{10}
\providecommand{\url}[1]{#1}
\csname url@samestyle\endcsname
\providecommand{\newblock}{\relax}
\providecommand{\bibinfo}[2]{#2}
\providecommand{\BIBentrySTDinterwordspacing}{\spaceskip=0pt\relax}
\providecommand{\BIBentryALTinterwordstretchfactor}{4}
\providecommand{\BIBentryALTinterwordspacing}{\spaceskip=\fontdimen2\font plus
\BIBentryALTinterwordstretchfactor\fontdimen3\font minus
  \fontdimen4\font\relax}
\providecommand{\BIBforeignlanguage}[2]{{%
\expandafter\ifx\csname l@#1\endcsname\relax
\typeout{** WARNING: IEEEtran.bst: No hyphenation pattern has been}%
\typeout{** loaded for the language `#1'. Using the pattern for}%
\typeout{** the default language instead.}%
\else
\language=\csname l@#1\endcsname
\fi
#2}}
\providecommand{\BIBdecl}{\relax}
\BIBdecl

\bibitem{estella2011icc}
I.~Estella and D.~G\"und\"uz, ``Expected distortion with fading channel and
  side information quality,'' in \emph{Proc. IEEE Int'l Conference on
  Communications (ICC)}, Jun. 2011, pp. 1--6.

\bibitem{Estella2011DistExponent}
------, ``{Distortion exponent in fading MIMO channels with time-varying side
  information},'' in \emph{Proc. IEEE Int'l Symposium on Information Theory
  Proceedings (ISIT),}, St. Petersburg, Russia, Aug. 2011, pp. 548 --552.

\bibitem{Estella2013SystematicLossy}
------, ``{Systematic lossy source transmission over Gaussian time-varying
  channels},'' in \emph{Proc. IEEE Int'l Symposium on Information Theory
  Proceedings (ISIT)}, Istanbul, Turkey, Jul. 2012.

\bibitem{wyner1978rate}
A.~Wyner, ``{The rate-distortion function for source coding with side
  information at the decoder},'' \emph{Information and Control}, vol.~38,
  no.~1, pp. 60--80, Jan. 1978.

\bibitem{Shamai:IT:98}
S.~Shamai, S.~Verd\'{u}, and R.~Zamir, ``Systematic lossy source-channel
  coding,'' \emph{IEEE Trans. on Information Theory}, vol.~44, no.~2, pp.
  564--579, Mar. 1998.

\bibitem{mittal2002hybrid}
U.~Mittal and N.~Phamdo, ``{Hybrid digital--analog (HDA) joint source--channel
  codes for broadcasting and robust communications},'' \emph{IEEE Trans. on
  Information Theory}, vol.~48, no.~5, May 2002.

\bibitem{Goblick:IT:65}
T.~J. Goblick, ``Theoretical limitations on the transmission of data from
  analog sources,'' \emph{IEEE Trans. on Information Theory}, vol.~11, no.~11,
  pp. 558--567, Nov. 1965.

\bibitem{Gastpar:IT:03}
M.~Gastpar, B.~Rimoldi, and M.~Vetterli, ``To code, or not to code: Lossy
  source-channel communication revisited,'' \emph{IEEE Trans. on Information
  Theory}, vol.~49, no.~5, pp. 1147--1158, May 2003.

\bibitem{tuncel2006slepian}
E.~Tuncel, ``{Slepian-Wolf coding over broadcast channels},'' \emph{IEEE Trans.
  on Information Theory}, vol.~52, no.~4, pp. 1469--1482, Apr. 2006.

\bibitem{Nayak2010DigSchemes}
J.~Nayak, E.~Tuncel, and D.~Gunduz, ``Wyner--{Z}iv coding over broadcast
  channels: {D}igital schemes,'' \emph{IEEE Trans. on Information Theory},
  vol.~56, no.~4, pp. 1782--1799, Apr. 2010.

\bibitem{Gunduz2013:Relay}
D.~Gunduz, E.~Erkip, A.~Goldsmith, and H.~Poor, ``Reliable joint source-channel
  cooperative transmission over relay networks,'' \emph{IEEE Trans. on
  Information Theory}, vol.~59, no.~4, pp. 2442--2458, Apr. 2013.

\bibitem{wilson2010joint}
M.~Wilson, K.~Narayanan, and G.~Caire, ``{Joint source channel coding with side
  information using hybrid digital analog codes},'' \emph{IEEE Trans. on
  Information Theory}, vol.~56, no.~10, pp. 4922--4940, Oct. 2010.

\bibitem{Huang2012CorrInt}
Y.-C. Huang and K.~Narayanan, ``Joint source-channel coding with correlated
  interference,'' \emph{IEEE Trans. on Communications}, vol.~60, no.~5, pp.
  1315--1327, May 2012.

\bibitem{Varasteh2012:SP}
M.~Varasteh and H.~Behroozi, ``{Optimal HDA schemes for transmission of a
  Gaussian source over a Gaussian channel with bandwidth compression in the
  presence of an interference},'' \emph{IEEE Trans. on Signal Processing},
  vol.~60, no.~4, pp. 2081--2085, Apr. 2012.

\bibitem{lapidoth2010sending}
A.~Lapidoth and S.~Tinguely, ``{Sending a bivariate Gaussian over a Gaussian
  MAC},'' \emph{IEEE Trans. on Information Theory}, vol.~56, no.~6, pp.
  2714--2752, Jun. 2010.

\bibitem{Estella2015ITW:JSCCIC}
I.~Estella and D.~G\"und\"uz, ``Gaussian joint source-channel coding for the
  strong interference channel,'' in \emph{Proc. IEEE Information Theory
  Workshop (ITW)}, Jerusalem, Israel, May. 2015, pp. 277--281.

\bibitem{Tian2011BivariateBroadcast}
C.~Tian, S.~Diggavi, and S.~Shamai, ``The achievable distortion region of
  sending a bivariate {G}aussian source on the {G}aussian broadcast channel,''
  \emph{IEEE Trans. on Information Theory}, vol.~57, no.~10, pp. 6419--6427,
  Oct. 2011.

\bibitem{ng2007minimumLayered}
C.~Ng, D.~G\"und\"uz, A.~Goldsmith, and E.~Erkip, ``{Distortion minimization in
  Gaussian layered broadcast coding with successive refinement},'' \emph{IEEE
  Trans. on Information Theory}, vol.~55, no.~11, pp. 5074--5086, Nov. 2009.

\bibitem{tian2008SuccRefinement}
C.~Tian, A.~Steiner, S.~Shamai, and S.~Diggavi, ``{Successive refinement via
  broadcast: Optimizing expected distortion of a Gaussian source over a
  Gaussian fading channel},'' \emph{IEEE Trans. on Information Theory},
  vol.~54, no.~7, pp. 2903 --2918, Jul. 2008.

\bibitem{gunduz2008joint}
D.~G\"und\"uz and E.~Erkip, ``{Joint source--channel codes for MIMO
  block-fading channels},'' \emph{IEEE Trans. on Information Theory}, vol.~54,
  no.~1, pp. 116--134, Jan. 2008.

\bibitem{Caire2007hybrid}
G.~Caire and K.~Narayanan, ``{On the distortion SNR exponent of hybrid
  digital--analog space--time coding},'' \emph{IEEE Trans. on Information
  Theory}, vol.~53, no.~8, pp. 2867--2878, Aug. 2007.

\bibitem{bhattad2008distortion}
K.~Bhattad, K.~Narayanan, and G.~Caire, ``{On the distortion SNR exponent of
  some layered transmission schemes},'' \emph{IEEE Trans. on Information
  Theory}, vol.~54, no.~7, pp. 2943--2958, Jul. 2008.

\bibitem{ng2007minimum}
C.~Ng, C.~Tian, A.~Goldsmith, and S.~Shamai, ``{Minimum expected distortion in
  Gaussian source coding with uncertain side information},'' in \emph{Proc.
  IEEE Information Theory Workshop (ITW)}, Sep. 2007, pp. 454--459.

\bibitem{Zhao2010ImpactSideInfo}
S.~Zhao, R.~Timo, T.~Chan, A.~Grant, and D.~Tuninetti, ``The impact of
  side-information on {Gaussian} source transmission over block-fading
  channels,'' in \emph{Proc. IEEE Int'l Conference on Communications (ICC)},
  May 2010, pp. 1--5.

\bibitem{Koken2013LossyHDA}
E.~Koken and E.~Tuncel, ``Gaussian {HDA} coding with bandwidth expansion and
  side information at the decoder,'' in \emph{Proc. IEEE Int'l Symposium on
  Information Theory (ISIT)}, Jul. 2013, pp. 11--15.

\bibitem{Berger1985SideInfMayBeAbsent}
C.~Heegard and T.~Berger, ``Rate distortion when side information may be
  absent,'' \emph{IEEE Trans. on Information Theory}, vol.~31, no.~6, pp. 727
  -- 734, Nov. 1985.

\bibitem{Steinberg2006HierarchicalJointSourceChannel}
Y.~Steinberg and N.~Merhav, ``On hierarchical joint source-channel coding with
  degraded side information,'' \emph{IEEE Trans. on Information Theory},
  vol.~52, no.~3, pp. 886--903, Mar. 2006.

\bibitem{elGamal:book}
A.~E. Gamal and Y.-H. Kim, \emph{Network Information Theory}.\hskip 1em plus
  0.5em minus 0.4em\relax Cambridge University Press, 2011.

\bibitem{Cover:book}
T.~M. Cover and J.~A. Thomas, \emph{Elements of Information Theory}.\hskip 1em
  plus 0.5em minus 0.4em\relax Wiley-Interscience, 1991.

\bibitem{Minero2015:Hybrid}
P.~Minero, S.~H. Lim, and Y.-H. Kim, ``A unified approach to hybrid coding,''
  \emph{IEEE Trans. on Information Theory}, vol.~61, no.~4, pp. 1509--1523,
  Apr. 2015.

\bibitem{Gaspart2008Uncoded}
M.~Gastpar, ``Uncoded transmission is exactly optimal for a simple {G}aussian
  sensor network,'' \emph{IEEE Trans. on Information Theory}, vol.~54, no.~11,
  pp. 5247 --5251, Nov. 2008.

\bibitem{Laneman2005SourceChannelDiv}
J.~Laneman, E.~Martinian, G.~W. Wornell, and J.~Apostolopoulos,
  ``Source-channel diversity for parallel channels,'' \emph{IEEE Trans. on
  Information Theory}, vol.~51, no.~10, pp. 3518--3539, Oct. 2005.

\bibitem{Dembo:book}
A.~Dembo and O.~Zeitouni, \emph{Large Deviations Techniques and Applications},
  2nd~ed.\hskip 1em plus 0.5em minus 0.4em\relax Springer, New York, 1998.

\bibitem{zheng2003diversity}
L.~Zheng and D.~Tse, ``{Diversity and multiplexing: A fundamental tradeoff in
  multiple-antenna channels},'' \emph{IEEE Trans. on Information Theory},
  vol.~49, no.~5, pp. 1073--1096, May. 2003.

\end{thebibliography}

\end{document}